\documentclass{article}

\usepackage[english]{babel}

\usepackage[a4paper, top=3cm, bottom=3cm, left=3cm, right=3cm]
{geometry}
\usepackage{graphicx}
\usepackage{array} 
\usepackage{adjustbox}
\usepackage{amsmath}
\usepackage{amsthm}
\usepackage{amssymb}
\usepackage{tabularx}
\usepackage{rotating}

\usepackage{booktabs}
\usepackage{graphicx}
\usepackage{xcolor}
\usepackage{multirow}
\usepackage{multicol}
\usepackage{bigstrut}
\usepackage{algorithmic}
\usepackage[ruled]{algorithm2e}
\usepackage[colorlinks=true, allcolors=blue]{hyperref}
\usepackage{mfirstuc} 
\usepackage{float}
\usepackage{makecell}
\usepackage{appendix}
\usepackage{arydshln}
\usepackage{soul}

\newcolumntype{P}[1]{>{\centering\arraybackslash}p{#1}}

\newtheorem{theorem}{Theorem}
\newtheorem{corollary}{Corollary}
\newtheorem{lemma}{Lemma}
\newtheorem{definition}{Definition}

\newtheorem{observation}{Observation}
\newtheorem{claim}{Claim}

\newcommand{\ssc}[1]{\textsc{\capitalisewords{\MakeLowercase{#1}}}}

\newcommand{\cI}{\mathcal{I}}
\newcommand{\OPT}{\mathsf{OPT}}
\newcommand{\SC}{\mathsf{SC}}
\newcommand{\SW}{\mathsf{SW}}
\newcommand{\CoF}{\mathsf{CoF}}

\begin{document}
\title{A Fair Allocation is Approximately Optimal for Indivisible Chores, or Is It?\thanks{The authors are ordered alphabetically. Ankang Sun is the corresponding author.}}
%
%
\author{Bo Li\thanks{Department of Computing, The Hong Kong Polytechnic University, China. comp-bo.li@polyu.edu.hk} \and
Ankang Sun\thanks{Department of Computing, The Hong Kong Polytechnic University, China. ankang.sun@polyu.edu.hk} \and
Shiji Xing\thanks{Department of Computing, The Hong Kong Polytechnic University, China. shi-ji.xing@connect.polyu.hk}}
%
%
%
\date{}
\maketitle              
\begin{abstract}

In this paper, we study the allocation of indivisible chores and consider the problem of finding a fair allocation that is approximately efficient.
We shift our attention from the multiplicative approximation to the additive one. 
Our results are twofold, with (1) bounding how the optimal social cost escalates resulting from fairness requirements and (2) presenting the hardness of approximation for the problems of finding fair allocations with the minimum social cost. 
To quantify the escalation, we introduce {\em cost of fairness} (CoF) --- an alternative to the price of fairness (PoF) --- to bound the {\em difference} (v.s. {\em ratio} for PoF) between the optimal social cost with and without fairness constraints in the worst-case instance.
We find that CoF is more informative than PoF for chores in the sense that the PoF is infinity regarding all EQX (equitable up to any item), EQ1 (equitable up to one item) and EF1 (envy-free up to one item), while the CoF is $n$ regarding EQX and 1 regarding EQ1 and EF1, where $n$ is the number of agents.
For inapproximability, we present a detailed picture of hardness of approximation. 
We prove that finding the optimal EQX allocation within an additive approximation factor of $n$ is NP-hard for any $n\ge 2$ where $n$ is the number of agents and the cost functions are normalized to 1. 
For EQ1 and EF1, the problem is NP-hard when the additive factor is a constant and $n\ge 3$.
When $n=2$, we design additive approximation schemes for EQ1 and EF1.
\end{abstract}
\section{Introduction}
The allocation of resources among multiple agents has long been a fundamental problem in various economic contexts \cite{DBLP:journals/ai/AmanatidisABFLMVW23,DBLP:books/daglib/0017730,steinhaus1948problem,DBLP:reference/choice/Thomson16}. Two orthogonal objectives, social welfare and individual fairness, used to be studied independently. 
A recent research trend focuses on understanding the trade-off between them \cite{DBLP:journals/ior/BertsimasFT11,DBLP:conf/wine/CaragiannisKKK09,DBLP:journals/mst/CaragiannisKKK12}.
There are two types of problem. 
The first one is more quantitative, where we want to understand how much efficiency would be lost due to fairness constraints.
The second one is more computational, where we want to design polynomial time
algorithms to compute the optimal fair allocation, i.e., the fair allocation with maximum social welfare. 
In this work, we align with this research trend and address these two problems concerning the allocation of indivisible chores when agents incur costs to carry out their assigned chores.
 
Two widely studied comparison-based fairness notions are {\em equitability} (EQ) \cite{DBLP:books/daglib/0017730} and  {\em envy-freeness} (EF) \cite{GS58,Varian74}.
Informally, an allocation is EQ if all agents have the same cost and EF if no agent wants to exchange their bundles with any other agent's.
When the items are not divisible, an EQ or EF allocation cannot be guaranteed to exist even in straightforward cases such that there is a single item with non-zero cost. 
Accordingly, since \cite{Budish11}, most of the focus has been on their up to one relaxations. 
In particular, an allocation is {\em equitable up to one item} (EQ1) if every agent's cost excluding one item in her bundle is no greater than any other agent's cost, and is {\em envy-free up to one item}
(EF1) if every agent's cost excluding one item in her bundle is no greater than her cost on any other agent's bundle.
A stronger relaxation is up to any, which requires EQ and EF to be satisfied by each agent after the removal of any time in her bundle, and the resulting criteria are called EQX and EFX.
It is known that EQX, EQ1 and EF1 allocations are guaranteed to exist
\cite{DBLP:journals/teco/CaragiannisKMPS19,DBLP:conf/ecai/GourvesMT14,DBLP:conf/sigecom/LiptonMMS04}, but the existence of EFX is still unknown.

To quantify the increase in social cost by restricting allocations to be fair, Caragiannis et al. \cite{DBLP:conf/wine/CaragiannisKKK09,DBLP:journals/mst/CaragiannisKKK12} introduced the concept of {\em price of fairness} (PoF), which is the worst-case ratio between the optimal social cost with and without fairness constraints. 
Their original work focused on the unrelaxed fairness criteria EQ and EF, and excluded the instances which do not admit these fair allocations.
Later, Bei et al. \cite{DBLP:journals/mst/BeiLMS21} and other followup works \cite{DBLP:conf/wine/BarmanB020,DBLP:conf/www/0037L022,DBLP:conf/atal/LiLLTT24,DBLP:journals/aamas/SunCD23,DBLP:journals/aamas/SunCD23a} studied the price of fairness regarding guaranteed fairness such as EQ1 and EF1. 
When the items are indivisible chores, it is known that there are instances whose optimal social cost is close to 0 but the social cost of any fair (in terms of all EQX, EQ1 and EF1) allocation is a constant, resulting in unbounded PoF.
These results indicate that in the context of indivisible chores, fairness and efficiency are not compatible.
Following these results, in this work, we want to understand the following problem.

\begin{quote}
    \em If we could tolerate a constant increase in social cost, supposing that the total cost of items for every agent is normalized to 1, 
    can we ensure the allocations are fair?
\end{quote}

 To investigate the above problem, we propose a new quantitative measure,
the {\em cost of fairness} (CoF), {which is} the {\em difference}, in the worst-case instances, between the optimal social cost under fairness constraints and the optimal social cost without constraints.
As we will prove in this work, CoF is more precise and revealing than PoF in terms of quantifying the efficiency loss in the context of indivisible chores.
The known results that the PoF is infinity for all EQX, EQ1, and EF1 do not answer the above question nor
provide insights into whether certain fairness criteria can be guaranteed while causing less social cost escalation than others.
In stark contrast, the CoF is $n$ regarding EQX and 1 regarding EQ1 and EF1.
This exactly answers our problem: if we could tolerate a social cost loss of 1, we can ensure the allocations are EQ1 or EF1; however, EQX cannot be ensured.

We then take a computation perspective to study the previous problem.
Among all fair allocations, we want to compute the one that is the most efficient.
This problem has been extensively studied for indivisible goods, particularly under fairness notions of EF1 and EQ1  \cite{DBLP:conf/atal/BarmanG0KN19,DBLP:journals/eor/AzizHMS23,DBLP:journals/corr/abs-2205-14296,DBLP:journals/aamas/SunCD23}.
These results show that such a problem is NP-hard and is even hard to approximate. 
For indivisible chores, although the problem is proved to be NP-hard for EQ1 and for EQX in \cite{DBLP:journals/aamas/SunCD23}, the hardness of approximation is still unknown. 
Moreover, we are not aware of any prior work that studied the problem for EF1.
Therefore, our paper also aims to fill the gap regarding the hardness of approximation.
Inspired by the additivity nature of the cost of fairness, we investigate the computational complexity of finding an approximate solution within additive factors.

\subsection{Our Contribution}
In this work, we propose to study the {\em cost of fairness} for the allocation of indivisible chores and consider three guaranteed criteria, namely, EQX, EQ1, and EF1.
On one hand, we give asymptotically tight bounds of the cost of fairness regarding all three fairness notions.
On the other hand, we consider the computation problem of finding the fair allocations with minimum social cost.
We show that all these problems are NP-hard, and are also hard to be approximated within 
either a constant or a linear additive factor. 
Our main results and their comparison to the price of fairness are shown in Table \ref{tab:main-results}.
If not stated otherwise, every agent's cost on all items is normalized to 1.
In the following, we give a detailed introduction of our main results.

\begin{table}[]
    \centering
    \caption{Main results. PoF and CoF respectively refer to price of fairness and cost of fairness. Th.$x$ points to Theorem $x$. The results on the price of fairness are from \cite{DBLP:journals/aamas/SunCD23} and \cite{DBLP:journals/aamas/SunCD23a}.
    The parameter $n$ is the number of agents and ``$< \OPT+c$'' means the problem of finding a fair allocation (EQX, EQ1, and EF1 respectively) with social cost smaller than the social cost of the optimal fair allocation plus $c$.}
    \label{tab:main-results}
    \renewcommand{\arraystretch}{2}
    \begin{tabular}{|c|c|c|c|}
    \hline
        Fairness & PoF & CoF & Minimizing Social Cost \\
        \hline
        EQX & $\infty$ & $n$ (Th. \ref{thm::cof-eqx}) & \makecell{$n\ge 2$: $< \OPT+n$ is NP-hard (Th. \ref{thm:hard-eqx-sc}) }\\
        \hline
        EQ1 & $\infty$ & 1 (Th. \ref{thm::eq1-sc-at-most-1}) & \makecell{$n\ge 3$: $< \OPT+\frac{1}{2}$ is NP-hard (Th. \ref{thm:hard-eq1-sc}) \\ $n=2$: additive approximation scheme (Th. \ref{thm:eq1-ptas}) 
        }  \\
        \hline
        EF1 & $\infty$ & 1 (Th. \ref{thm::ef1-sc-at-most-1}) & \makecell{$n\ge 3$: $< \OPT+\frac{n-1}{2n+2}$ is NP-hard (Th. \ref{thm:hard-ef1-sc}) \\ $n=2$: additive approximation scheme (Th. \ref{thm:ef1-ptas}) 
        }
        \\
        \hline
    \end{tabular}
\end{table}

\paragraph{Result 1.} 
We start with EQX allocations, and our results in this part are negative.
We first prove that the cost of fairness regarding EQX is $n$, which means that there are instances 
for which every EQX allocation results in a social cost as unfavorable as when each agent is assigned all the chores.
This result further strengthens the hard instance designed in \cite{DBLP:journals/aamas/SunCD23}, where the optimal social cost is close to zero but the EQX allocation has a small constant cost resulting the price of fairness being unbounded.
We then study the optimization problem of finding the EQX allocations with the minimum social cost.
We prove that computing an approximate solution with an additive approximation factor smaller than $n$ is NP-hard.
As each agent's cost on all items is normalized to 1,
our results are tight in the following two perspectives; (1) any EQX allocation (approximately) achieves the minimum social cost in the worst case instance,
and (2) any EQX allocation achieves the best possible additive (polynomial time) approximation. if P $\neq$ NP.


\paragraph{Result 2.} 
We next consider EQ1 allocations and our results are significantly more positive. 
In sharp contrast to EQX, we first prove that every fair allocation instance admits an EQ1 allocation (which can be found in polynomial time) whose social cost is no greater than 1 -- a single agent's total cost. 
This implies that the cost of EQ1 is upper bounded by 1.
We also prove that the bound is asymptotically tight. 
This result is noteworthy because in the worst-case instance, the cost of an allocation can be as large as $n$.
Therefore, the actual distance between the most efficient allocation and the most efficient EQ1 allocation can be relatively small, which cannot be concluded from the price of EQ, EQ1 or the price of EQX (which are  all infinite as proved in \cite{DBLP:journals/mst/CaragiannisKKK12} and \cite{DBLP:journals/aamas/SunCD23}). 
We then study the optimization problem of finding the EQ1 allocations with the minimum social cost. 
This problem has been proved to be NP-hard, even when $n=2$  \cite{DBLP:journals/aamas/SunCD23}.
We improve this result and show that for $n\ge 3$, there is no polynomial-time algorithm that ensures an additive approximation factor smaller than $\frac{1}{2}$ given P $\neq$ NP.
When $n=2$, we complement the hardness results with an additive approximation scheme\footnote{The additive approximation scheme aims to compute a solution with an absolute error of at most $\epsilon h$ in the objective where $\epsilon$ is the accuracy parameter and $h$ is a predetermined fixed parameter.
The additive approximation scheme has been applied to study load balancing problems \cite{DBLP:conf/icalp/BuchemRVW21}.}.


\paragraph{Result 3.} 
We also study EF1 allocations. 
Similar to EQ1, we prove that the asymptotically tight bound of the cost of EF1 is 1. Moreover, the fair allocation with the matched social cost
can be found in polynomial time.
Note that the price of EF and that of EF1 are infinite as proved in  \cite{DBLP:journals/mst/CaragiannisKKK12} and \cite{DBLP:journals/aamas/SunCD23}.
Regarding the optimization problem, we prove that finding an EF1 allocation with the minimum social cost is NP-hard and moreover, when $n\ge 3$, there is no polynomial-time algorithm that ensures an additive approximation factor smaller than $\frac{n-1}{2n+2}$ given P $\neq$ NP.
Note that $\frac{n-1}{2n+2}\ge \frac{1}{4}$ and goes to $\frac{1}{2}$ when $n$ goes to infinity.  
When $n=2$, we also design an additive approximation scheme.

It is worthwhile to mention that we en-route prove that for every optimization problem we consider, there is no polynomial-time algorithm having a bounded multiplicative approximation guarantee. This fact also argues in favor of using cost of fairness to investigate the fairness-efficiency trade-off for chores.


\subsection{Related Work}
\label{app:related-work}

The concept of the price of fairness draws inspiration from the extensively studied notions of the price of anarchy \cite{DBLP:conf/stacs/KoutsoupiasP99} and the price of stability \cite{DBLP:journals/siamcomp/AnshelevichDKTWR08} in game theory.
Caragiannis et al. \cite{DBLP:journals/mst/CaragiannisKKK12} first considered the price of EF, EQ and PROP (short for proportionality \cite{steinhaus1948problem}) for divisible and indivisible goods and chores.
Bei et al. \cite{DBLP:journals/mst/BeiLMS21} and followup works \cite{DBLP:conf/wine/BarmanB020,DBLP:conf/www/0037L022,DBLP:conf/atal/LiLLTT24,DBLP:journals/aamas/SunCD23,DBLP:journals/aamas/SunCD23a} propose to study the price of the up to one relaxation of these fairness notions.
Barman et al. \cite{DBLP:conf/wine/BarmanB020} also considered the price of approximate maximin share fairness, namely $\frac{1}{2}$-MMS. 
Some recent works study the price of fairness when items have graph structures \cite{DBLP:journals/iandc/HohneS21,DBLP:journals/dam/Suksompong19,DBLP:conf/ecai/SunL23,DBLP:conf/atal/Sun024}.
In the graphical setting, the paradigm of quantifying welfare loss under constraints is also applied to the requirement of connected allocation to each agent, i.e., price of connectivity 
\cite{DBLP:journals/siamdm/BeiILS22,DBLP:journals/corr/abs-2405-03467}.

In addition to social welfare, Pareto optimality (PO) is also widely studied in resource allocation to measure the efficiency of an allocation.
Caragiannis et al. \cite{DBLP:journals/teco/CaragiannisKMPS19} first observed that for indivisible goods, any allocation that maximizes the Nash welfare is EF1 and thus proved the existence of an EF1 and PO allocation. 
Later, Barman et al. \cite{DBLP:conf/sigecom/BarmanKV18} designed a pseudopolynomial time algorithm for this problem.
It is still unknown whether an EF1 and PO allocation can be found in strongly polynomial time.
For chores, the compatibility between EF1 and PO is largely unknown, except for several special cases \cite{DBLP:conf/atal/EbadianP022,DBLP:conf/aaai/GargMQ22,DBLP:journals/corr/abs-2402-17173,DBLP:conf/sigecom/0001Z023}.
A PROP1 and PO allocation exists and can be found in  strongly polynomial time for goods, chores and even a mixture of them \cite{DBLP:journals/orl/AzizMS20,DBLP:conf/aaai/BarmanK19}.
However, it is shown in \cite{DBLP:journals/orl/AzizMS20} that a PROPX allocation may not exist for goods.
Although PROPX allocations exist for chores, it may not be compatible with PO \cite{DBLP:journals/orl/AzizMS20} when the excluded item can have zero cost.
It is an open problem when the excluded item must have positive cost.
Finally, EQ1 is compatible with PO for chores but not for goods, and EQX is not for both goods and chores \cite{DBLP:conf/ijcai/FreemanSVX19,DBLP:conf/atal/FreemanSVX20}.
It is NP-hard to decide whether an instance admits an EQX (or EQ1 for goods) and PO allocation.

\section{Preliminaries}
Given a positive integer $k$, denote by $[k]$ the set $\{1,\ldots,k\}$. There is a set of indivisible items $O = \{ o_1,\ldots, o_m \}$ to be allocated to a set of agents $N = \{ 1,\ldots, n \}$. 
In this paper, we reserve $n$ and $m$ for the number of agents and items.
Each agent $i \in N$ is associated with a disutility or cost function $c_i: 2^O \rightarrow \mathbb{R}_{\geq 0}$. For simplicity, we write $c_i(o)$ to refer $c_i(\{ o \})$, i.e., the cost of $o$ for agent $i$.
It is assumed that the cost functions are normalized and additive, i.e., $c_i(O) = 1$ and for any $S\subseteq O$, $c_i(S) = \sum_{o\in S}c_i(o)$.
An instance of the allocation problem is denoted by $\cI = \langle N,O,\{c_i\}_{i\in N} \rangle$.

An allocation $A = (A_1,\ldots,A_n)$ is an ordered $n$-partition of $O$ where for any $i \neq j$, $A_i\cap A_j =\emptyset$ and $\bigcup_{i \in N} A_i = O$. 
The \emph{social cost} of an allocation is the total cost of agents. Given an allocation $A$, the social cost of $A$ is $\SC(A) = \sum_{i\in [n]} c_i(A_i)$. 
The allocation with minimum social welfare without fairness constraints is called {\em optimal}.
The allocation with the minimum social welfare among all allocations satisfying fairness criterion P is called an optimal P allocation.

In this paper, we consider two comparison-based fairness criteria, equitability (EQ) and envy-freeness (EF).
Given an allocation $A = (A_1,\ldots,A_n)$, $A$ is EQ if $v_i(A_i) = v_j(A_j)$ for all $i,j \in N$, and is EF if $v_i(A_i) \ge v_i(A_j)$ for all $i,j \in N$.
Since an EQ or EF allocation is not guaranteed to exist, we consider their up to one relaxations. 

\begin{definition}[Equitable up to Any Item]
    An allocation $A=(A_1, \ldots, A_n)$ is {\em equitable up to any item} (EQX) if for any pair of agents $i,j\in N$ and for every $o\in A_i$, $c_i(A_i \setminus \{o\}) \leq c_j(A_j)$. 
\end{definition}

\begin{definition} [Equitable up to One Item]
     An allocation $A=(A_1, \ldots, A_n)$ is {\em equitable up to one item} (EQ1) if for any agents $i,j\in N$, there exists some $o\in A_i$ such that $c_i(A_i \setminus \{o\}) \leq c_j(A_j)$. 
\end{definition}


\begin{definition}[Envy-free up to One Item]
    An allocation $A=(A_1, \ldots, A_n)$ is {\em envy-free up to one item} (EF1) if for any $i, j \in N$, there exists an item $o\in A_i$ such that $c_i(A_i\setminus\{ o \}) \leq c_i(A_j)$. 
\end{definition}




We now introduce the concept of \emph{cost of fairness} (CoF). 
Informally, given a fairness criterion, the CoF is the supreme difference, among all instances, between the minimum social cost of all fair allocations and that of all allocations.

\begin{definition}
    For a given fairness criterion $P$, the cost of fairness ($\CoF$) regarding $P$, also called the cost of $P$, is defined as 
    $$
    \sup\limits_{\cI}\min_{A\in P(\cI)} \SC(A) - \OPT(\cI),
    $$
    where $P(\cI)$ is the set of allocations satisfying $P$ and $\OPT(\cI)$ is the minimum social cost over all allocations of $\cI$.
\end{definition}

\section{The Social Cost of EQX Allocations}\label{sec:EQX}

Our results in this section are negative. 
We first prove that in the worst case example, any EQX allocation can be as bad as the ``worst allocation'' when every agent gets a cost of 1, while the social cost of optimal solution is close to 0,  implying the cost of EQX being $n$.
To surpass this obstacle, we consider the computation problem of finding the EQX allocation with the minimum cost.
Unfortunately, we prove that for any $n\geq 2$, no polynomial time algorithm ensures an EQX allocation with social cost less than $n$, unless P = NP.

\subsection{Bounding the Cost of EQX}
We first bound the cost of EQX. 
 

\begin{theorem}\label{thm::cof-eqx}
    For any $n\geq 2$, the cost of EQX is $n$.
\end{theorem}
\begin{proof}
    Since the social cost of any allocation is at most $n$ and EQX allocation always exists \cite{DBLP:conf/ecai/GourvesMT14}, the cost of EQX is at most $n$.

    As for the lower bound, let us consider an instance with $n$ agents and $2n-1$ items. The items are divided into two categories, a set $\{o^1_1,\ldots,o^1_n\}$ of $n$ items for the first category and a set $\{o^2_1,\ldots,o^2_{n-1}\}$ of $n-1$ items for the second category. For simple notations, let $O^1=\{o^1_1,\ldots,o^1_n\}$ and $O^2=\{o^2_1,\ldots,o^2_{n-1}\}$.
    The cost functions of agents are presented in Table \ref{tab:eqx-worse-case-example} where $K$ is a sufficiently large number, i.e., $K\gg n2^n$. 

    \begin{table}[ht]
  \centering
  \caption{The constructed instance for Theorem \ref{thm::cof-eqx}}
  \label{tab:eqx-worse-case-example}
    \begin{tabular}{|c|c|cccccccccccc|}
    \hline
    \multicolumn{2}{|c|}{\multirow{2}[4]{*}{$c_i(o)$}} & \multicolumn{12}{c|}{Items} \\
\cline{3-14}    \multicolumn{2}{|c|}{} & $o^1_1$ & $o^1_2$ & $o^1_3$ & $\hdots$ & $\hdots$ & $o^1_n$ & $o^2_1$ & $o^2_2$ & $o^2_3$ & $\hdots$ & $\hdots$ & $o^2_{n-1}$ \bigstrut\\
    \hline
    \multirow{6}[2]{*}{Agents} & $1$   & $K$   & $1$   & $1$   & $\hdots$ & $\hdots$ & $1$   & $n$   & $2n$  & $4n$  & $\hdots$ & $\hdots$ & $n\cdot2^{(n-2)}$\\
          & $2$   & $1$   & $K$   & $1$   & $\hdots$ & $\hdots$ & $1$   & $n$   & $2n$  & $4n$  & $\hdots$ & $\hdots$ & $n\cdot2^{(n-2)}$ \\
          & $3$   & $1$   & $1$   & $K$   & $\hdots$ & $\hdots$ & $1$   & $n$   & $2n$  & $4n$  & $\hdots$ & $\hdots$ & $n\cdot2^{(n-2)}$ \\
          & $\vdots$ & $\vdots$ & $\vdots$ & $\vdots$ & $\ddots$ &        & $\vdots$ & $\vdots$ & $\vdots$ & $\vdots$ & {$\ddots$} & & $\vdots$ \\
          & $\vdots$ & $\vdots$ & $\vdots$ & $\vdots$ &       & $\ddots$ & $\vdots$ & $\vdots$ & $\vdots$ & $\vdots$ & &        $\ddots$& $\vdots$ \\
          & $n$   & $1$   & $1$   & $1$   & $\hdots$ & $\hdots$ & $K$   & $n$   & $2n$  & $4n$  & $\hdots$ & $\hdots$  & $n\cdot2^{(n-2)}$ \\
    \hline
    \end{tabular}
\end{table}%
    
  

    For any $i\in[n]$, the total cost of items for agent $i$ is $K+n2^{n-1}-1$ and hence cost functions are normalized. 
    Note that by scaling, these cost functions can be converted to cost functions with total cost 1. For the ease of representation, throughout this paper, when presenting cost functions, we do not scale. 
    By allocating each item to the agent having the least cost for it, one can verify that
    the social cost of optimal allocation
    is equal to $n2^{n-1}$.
    We below show that any EQX allocation has a social cost of at least $nK$, 
    that is, every agent receives the item with cost $K$ for her in any EQX allocation.
    
    For a contradiction, assume an EQX allocation $A$ exists such that $\SC(A)<nK$. Then there exists at least some agent $i$ with $c_i(A_i)<K$.
    We next divide agents into two categories $N_1$ and $N_2$; for any $i\in N_1$, $c_i(A_i) <K$; for any $i\in N_2$, $c_i(A_i) \geq K$. In other words, every agent $i\in N_1$ does not receive $o^1_i$ in $A$.
    Note that $N_1\neq \emptyset$ while $N_2$ could be empty. For any agent $i\in N_2$, $c_i(A_i)=K$ (equivalent to $A_i=\{o^1_i\}$) holds; otherwise, if $c_i(A_i)>K$, it is not hard to verify that $\max_{o\in A_i}c_i(A_i\setminus\{o\}) \geq K > c_j(A_j)$ for all $j\in N_1$, violating the assumption that $A$ is EQX. 
    Therefore for any $i\in N_2$, $A_i=\{o^1_i\}$, and moreover, all $o^2_j$'s are allocated among agents $N_1$.
    
    For $N_1$, we claim that there are at least two agents $i_1, i_2 \in N_1$ such that they get items from $O^1$ in $A$.
    First, if no agent from $N_1$ gets items from $O^1$, then items in $O^1$ are assigned to agents in $N_2$. As $A_i=\{o^1_i\}$ holds for all $i\in N_2$ and $|O^1|=n$, that means every agent $i$ belongs to $N_2$, implying $N_1=\emptyset$, a contradiction.
    If only one agent from $N_1$ receives items from $O^1$ and let agent $i_1 \in N_1$ be such an agent, then item $o^1_{i_1}$ is not allocated to agent $i_1$ as $c_{i_1}(A_{i_1}) < K$.
    Recall that for any $i\in N_2$, $A_i=\{o^1_i\}$ holds. Thus, $o_{i_1}^1$ is not allocated to any agent in $N_2$, and accordingly, there must be another agent (let us say $i_2$) from $N_1$ receiving $o^1_{i_1}$.

    We next discuss whether $A_{i_1}$ or $A_{i_2}$ contains items from $O^2$. There are three possible scenarios (without loss of generality):
     \begin{enumerate}
        \item $A_{i_1} \cap O^2 = \emptyset, A_{i_2} \cap O^2 = \emptyset$;
        \item $A_{i_1} \cap O^2 = \emptyset, A_{i_2} \cap O^2 \neq \emptyset$;
        \item $A_{i_1} \cap O^2 \neq \emptyset, A_{i_2} \cap O^2 \neq \emptyset$.
    \end{enumerate}
    Before the discussion over the cases, we first note that for any agent $i\in [n]$ and for any $ o^2_j \in O^2, c_i(o^2_j) = n 2^{j-1} > \sum_{t=1}^{j-1} c_i(o^2_t) + n-1 \geq n-1$.

    Now consider the first scenario.
    Since both $i_1,i_2$ do not receive items from $O^2$, all $o^2_j$'s are allocated among the remaining $|N_1|-2$ agents in $N_1$, meaning that there exists an agent $k$ who receives at least 2 items from $O^2$. 
    By $i_1, i_2 \in N_1$, $A_{i_1}\cap O^2=\emptyset$ and $A_{i_2}\cap O^2=\emptyset$, we have $\max\{ c_{i_1}(A_{i_1}), c_{i_2}(A_{i_2})\} \leq n -1 $.
    Since $\max_{o\in A_k} c_k(A_k \setminus \{o\}) \geq n$, allocation $A$ does not satisfy EQX, a contradiction. Thus, this scenario never happens.
    
   For the second scenario, as $A_{i_1}\cap O^2=\emptyset$, we have $c_{i_1}(A_{i_1})\leq  n - 1$.
   Note that $A_{i_2} \cap O^2 \neq \emptyset$ and $A_{i_2}\cap O^1 \neq \emptyset$, then it holds that  $\max_{o\in A_{i_2}} c_{i_2}(A_{i_2} \setminus \{o\}) \geq n > c_{i_1}(A_{i_1})$, implying that $A$ is not EQX, a contradiction. Thus, this scenario never happens.
    
    Consider the last scenario. Both $i_1$ and $i_2$ receive at least one item from $O^2$. Without loss of generality, assume $c_{i_1}(A_{i_1}) > c_{i_2}(A_{i_2})$. Recall that for any $i$ and any $o^2_j\in O^2$, $c_i(o^2_j) > \sum_{t=1}^{j-1}c_i(o^2_t) + n-1$ and the total cost of 
    $O^1\cap A_{i_2}$
    for agent $i_2$ is $c_{i_2}(A_{i_2}\cap O^1) \leq n-1$. 
    We have $\max_{o\in A_{i_1}} c_{i_1}(A_{i_1}\setminus\{o\}) \geq \max_{o\in A_{i_1}}c_{i_1}(o) > c_{i_2}(A_{i_2})$ and consequently allocation $A$ is not EQX.
    We again derive a contradiction.

    All possible scenarios lead to a contradiction and thus any EQX allocation should have a social cost of at least $nK$, that is, every agent $i$ receives item $o^1_i$.
    As the cost of agents for every item in $O^2$ is identical,
    the social cost of any EQX allocation is indeed equal to $nK+n2^{n-1} - n$.
    Therefore, by scaling the total cost of every agent to 1, we can claim that the cost of EQX is at least $$
    \frac{nK-n}{K+n2^{n-1}-1} \rightarrow n \textnormal{ as } K\rightarrow +\infty,
    $$
    which yields the desired lower bound.
\end{proof}



\subsection{Computing the Optimal EQX Allocations}
We now consider the problem of minimizing social cost subject to EQX constraints, denoted as SC-EQX. Although SC-EQX is known to be NP-hard \cite{DBLP:journals/aamas/SunCD23}, the hardness of approximation has not been studied.
Below we significantly strengthen the state-of-the-art result by proving an additive inapproximability factor of $n$.

To establish the hardness result, we provide a polynomial time reduction from the well-known \ssc{Partition} problem \cite{DBLP:books/fm/GareyJ79}: 
given a set $\{ p_1,\ldots, p_r \}$ of $r$ integers whose sum is $2T$, can $[r]$ be partitioned into two indices sets $I_1$ and $I_2$ such that $\sum_{ j\in I_1} p_j = \sum_{ j\in I_2} p_j = T$?

\begin{theorem}\label{thm:hard-eqx-sc}
    For any $n\geq 2$, there is no polynomial time algorithm approximating SC-EQX within an additive factor smaller than $n$ unless P = NP.
\end{theorem}


\begin{proof}
    We present the reduction from the \ssc{partition} Problem. Given an arbitrary instance of \ssc{partition} $\{ p_1,\ldots, p_r \}$, we construct a fair division $\cI$ with $n$ agents and $2n+r$ items.
    Items are divided into three categories, labelled by the superscript. 
    Let $O^1=\{o^1_1,\ldots,o^1_n\}$ be the set of the $n$ items belong to first category; $O^2=\{o^2_1,\ldots,o^2_r\}$ be the set of the $r$ items belong to the second category;
    $O^3=\{o^3_1,\ldots,o^3_n\}$ be the set of $n$ items belong to the third category. Agents' cost functions are described as follows:
    \begin{itemize}
        \item for $\{o^1_j\}$: every agent $i$ has cost $K-x_i$ for $o^1_i$ and cost 0 for each of other items;
        \item for $\{o^2_j\}$: for every agent $i\in \{1,2\}$, she has cost $p_j$ for $o^2_j$; for every agent $i\geq 3$, her cost for every $o^2_j$ is equal to $10r^{i-2}n^{(n+i-2)}T$;
        \item for $\{o^3_j\}$: for any $i\in \{1,2\}$, she has cost $4T$ for $o^3_i$ and cost $10n^{i-1}T$ for each of other items;
        for agent $i\geq 3$, she has cost $5T$ for $o^3_i$ and cost $10n^{i-1}T$ for each of other items.
    \end{itemize}
    In the construction, $K$ is a sufficiently large number, satisfying $K\gg T\cdot n^{2n}\cdot r^n$ and for any $i$, $x_i = c_i(O^2\cup O^3)$.
    The total cost of every agent is equal to $K$ and hence cost functions are normalized.
    The cost functions of agents are also presented in Table \ref{tab:eqx-hardness}.
    \begin{sidewaystable}[p]
  \centering
  \caption{The constructed instance for Theorem \ref{thm:hard-eqx-sc}}
  \label{tab:eqx-hardness}
    \begin{tabular}{|c|c|cccccc|}
    \hline
    \multicolumn{2}{|c|}{\multirow{2}[4]{*}{$c_i(o)$ for $O^1$}} & \multicolumn{6}{c|}{Items in $O^1$} \\
\cline{3-8}    \multicolumn{2}{|c|}{} & $o^1_1$ & $ o^1_2$ & $o^1_3$ & $\cdots$ & $o^1_{n-1}$ & $o^1_{n}$ 
\\
    \hline
    \multirow{6}[2]{*}{Agents} & $1$ & $K-x_1$ & $ 0$ & $0$ & $ \cdots $ & $  0 $ & $ 0 $   
    \\
          & $2$   & $0$ & $ K-x_2$ & $0$ & $ \cdots$ & $ 0$ & $ 0$   \\
          & $3$ & $0$ & $0$ & $K-x_3$ & $\cdots$ & $0$ & $0$   \\
          & $\vdots$ & $\vdots$ & $\vdots$ & $ \vdots$ & $\ddots$ & $ \vdots$ & $ \vdots$ \\
          & $n-1$ & $0$ & $0$ & $0$ & $\cdots$ & $K-x_{n-1}$ & $0$ \\
          & $n$ & $0$ & $0$ & $0$ & $\cdots$ & $0$ & $K-x_n$ \\
    \hline
     \multicolumn{2}{|c|}
    {\multirow{2}[4]{*}{$c_i(o)$ for $O^2$}} & \multicolumn{6}{c|}{Items in $O^2$} 
    \\
\cline{3-8}    \multicolumn{2}{|c|}{} & $o^2_1$ & $ o^2_2$ & $o^2_3$ & $\cdots$ & $o^2_{n-1}$ & $o^2_{n}$ \bigstrut\\
    \hline
    \multirow{6}[2]{*}{Agents} & $1$ & $p_1$ & $ p_2$ & $p_3$ & $ \cdots $ & $  p_{r-1} $ & $ p_r $   \bigstrut[t]\\
          & $2$   & $p_1$ & $ p_2$ & $p_3$ & $ \cdots$ & $ p_{r-1}$ & $ p_r$   \\
          & $3$ & $10rn^{n+1}T$ & $10rn^{n+1}T$ & $10rn^{n+1}T$ & $\cdots$ & $10rn^{n+1}T$ & $10rn^{n+1}T$   \\
          & $\vdots$ & $\vdots$ & $\vdots$ & $ \vdots$ & $\ddots$ & $ \vdots$ & $ \vdots$ \\
          & $n-1$ & $10r^{n-3}n^{2n-3}T$ & $10r^{n-3}n^{2n-3}T$ & $10r^{n-3}n^{2n-3}T$ & $\cdots$ & $10r^{n-3}n^{2n-3}T$ & $10r^{n-3}n^{2n-3}T$ \\
          & $n$ & $10r^{n-2}n^{2n-2}T$ & $10r^{n-2}n^{2n-2}T$ & $10r^{n-2}n^{2n-2}T$ & $\cdots$ & $10r^{n-2}n^{2n-2}T$ & $10r^{n-2}n^{2n-2}T$ \\
    \hline
       \multicolumn{2}{|c|}
    {\multirow{2}[4]{*}{$c_i(o)$ for $O^3$}} & \multicolumn{6}{c|}{Items in $O^3$} 
    \\
\cline{3-8}    \multicolumn{2}{|c|}{} & $o^3_1$ & $ o^3_2$ & $o^3_3$ & $\cdots$ & $o^3_{n-1}$ & $o^3_{n}$ \bigstrut\\
    \hline
    \multirow{6}[2]{*}{Agents} & $1$ & $4T$ & $ 10T$ & $10T$ & $ \cdots $ & $  10T $ & $ 10T $   \bigstrut[t]\\
          & $2$   & $10nT$ & $ 4T $ & $10nT$ & $ \cdots$ & $ 10nT$ & $ 10nT $   \\
          & $3$ & $10n^2T$ & $10n^2T$ & $5T$ & $\cdots$ & $10n^2T$ & $10n^2T$   \\
          & $\vdots$ & $\vdots$ & $\vdots$ & $ \vdots$ & $\ddots$ & $ \vdots$ & $ \vdots$ \\
          & $n-1$ & $10n^{n-2}T$ & $10n^{n-2}T$ & $10n^{n-2}T$ & $\cdots$ & $5T$ & $10n^{n-2}T$ \\
          & $n$ & $10n^{n-1}T$ & $10n^{n-1}T$ & $10n^{n-1}T$ & $\cdots$ & $10n^{n-1}T$ & $5T$ \\
    \hline
    \end{tabular}%
\end{sidewaystable}%

In what follows, we will show that when the \ssc{partition} instance has an answer ``yes'', there exists an EQX allocation with social cost $o(K)$. If \ssc{partition} instance has ``no'' answer, then the social cost every EQX allocation is close to $nK$.
The key ideas of the constructed fair instance are: in any EQX allocation, (i) if agent $i$ receives her \emph{big-item} $o^1_i$, then $o^1_i$ must be the only item in her bundle unless each of the other agents also receives their big-item, 
and (ii) if we require the total cost of the assignment of $O^1$ to be zero, then $O^1$ must be allocated to at least two agents, and moreover, their cost should be the same in order to satisfy EQX; note that when verifying EQX criterion, the item (if any) with cost 0 would be hypothetically removed.

Suppose we have a ``yes'' instance of \ssc{partition} and let $I_1$ and $I_2$ be a solution. Now we consider allocation $A=(A_1,\ldots, A_n)$ where
\begin{itemize}
    \item $A_1$ contains $o^1_2,\ldots
    ,o^1_{n}$ and $o^2_j$ for all $j\in I_1$ and $o^3_1$;
    \item $A_2$ contains $o^1_1$ and $o^2_j$ for all $j\in I_2$ and $o^3_2$.
    \item for any $i\geq 3$, $A_3$ contains $o^3_i$ only.
\end{itemize}
Informally, in allocation $A$, the total cost of allocating $O^1$ is 0 and $O^2$ is allocated based on the solution of the \ssc{partition} instance.
It is not hard to verify that for any $i\in[n]$, $c_i(A_i)=5T$. Hence, allocation $A$ is EQX and has a social cost $5nT$, equal to the social cost of the optimal allocation.

For the other direction, suppose we have a ``no'' instance of \ssc{partition}. 
We will prove that in any EQX allocation $B=(B_1,\ldots,B_n)$, every agent $i$ must receive her big-item $o^1_i$. Our first step is to show that in $B$, there exists at least one agent receiving her big-item.
For a contradiction, assume no agent receives their big-item, which implies that the cost caused by the assignment of $O^1$ is 0 in allocation $B$.
Let $N_1= \{ i: B_i\cap O^1 \neq \emptyset \}$. By the contradiction assumption, for any $i\in N_1$, $c_i(B_i\cap O^1) = 0$. Then it is not hard to verify that $|N_1| \geq 2$
because if one agent receives all $o^1_j$'s, she would receive her big-item.
For any $i,j\in N_1$, $c_i(B_i) = c_j(B_j)$ holds as when comparing EQX, 
item with cost 0 would be hypothetically removed. 
Moreover, for any $i\in N_1$, $c_i(B_i) >0$ holds; If $N_1=[n]$, then every agent need to have the same cost to satisfy EQX, and the above proposition holds trivially ; If $N_1\neq [n]$, then for every $j\in [n]
\setminus N_1$, agent $j$ can only receive one item in $B$ due to EQX property and there will be items left unallocated.
\begin{claim}\label{claim::eqx-proof-1}
    Neither of agent 1 nor agent 2 is in $N_1$.
\end{claim}
\begin{proof}[Proof of Claim \ref{claim::eqx-proof-1}]
Suppose not. We first consider the case where both agents 1 and 2 are in $N_1$. Then the condition of EQX requires $c_1(B_1)=c_2(B_2)$. Accordingly, $B_2\cap O^3 = \{o^3_2\}$ must hold since for any $j\neq 2$, $c_2(o^3_j) > c_1(O^2\cup O^3)$.
This then implies $B_1\cap O^3= \{o^3_1\}$ as for any $j\neq 1$, $c_1(o^3_j) > c_2(o^3_2) + \sum_{o\in O^2}c_2(o)$.
Since the \ssc{partition} instance has an answer ``no'', it is impossible to achieve $c_1(B_1) = c_2(B_2)$ when ensuring that for any $i\in [2]$, $B_i\cap O^3 = \{o^3_i\}$ and $c_i(B_i\cap O^1) = 0$, yielding the desired contradiction.

If only one of $\{1,2\}$ is in $N_1$,
we below prove the case where agent 1 is in $N_1$. A similar proof applies to the case when agent 2 is in $N_1$.
Since $|N_1| \geq 2$, $N_1$ also contains another agent $j \geq 3$. We now show that $c_1(B_1) = c_j(B_j)$ never holds.
In order to make the cost of agent $1$ be identical to that of agent $j$, it must hold that $B_j\cap O^3 = \{o^3_j\}$ and $B_j\cap O^2 = \emptyset$ as for any $o\in O^2 \cup O^3\setminus \{o^3_j\}$, $c_j(o)>c_1(O^2\cup O^3)$ holds.
As a consequence, we have $c_j(B_j) = 5T$. Note that \ssc{partition} instance has an answer ``no'' and $c_1(B_1\cap O^1) = 0$, then it is impossible to make $c_1(B_1)=5T$, yielding the desired contradiction.
\end{proof}

The above claim implies that every agent $i\in N_1$ satisfies $i\geq 3$. 
For any two agents $i,j\in N_1$ with $i<j$ (without loss of generality), the only possible way to achieve $c_i(B_i) = c_j(B_j) > 0 $ is that: agents $i,j$ do not receive items in $O^2$ and moreover $B_i\cap O^3 = \{o^3_i\}$ and $B_j\cap O^3 = \{o^3_j\}$.
Note that for any $o\in O^2$, agent $j$ has a cost $c_j(o)=10r^{j-2}n^{n+j-2} T \geq 10r^{i-1}n^{n+i-1} T $, larger than the total cost of $O^2\cup O^3$ for agent $i$. Thus, agent $j$ receives no item from $O^2$, which then implies that agent $i$ does not receive items from $O^2$ as for every $o\in O^2$, $c_2(o)=10r^{i-2}n^{n+i-2}$, larger than the total cost of $O^3$ for agent $j$.
Similarly, one can verify that in order to ensure $c_i(B_i)=c_j(B_j)$, agent $i$ (resp. agent $j$) receives no item from $O^3\setminus\{o^3_i\}$ (resp. $O^3\setminus\{o^3_j\}$).
Therefore, for every $i'\in N_1$, it must hold that $B_{i'}\cap O^3=\{o^3_{i'}\}$ and moreover $c_{i'}(B_{i'}) = 5T$ as $i' \in N_1$.
We then claim that for any $j\in [n]$, agent $j$ receives exactly one item from $O^3$. This is clearly true for agents in $N_1$.
As for $j\in [n]\setminus N_1$, if agent $j$ receives more than two items from $O^3$, she violates EQX when compared to every agent $i\in N_1$; if $j$ receives nothing from $O^3$, then another agent $j'$ would receive at least two items from $O^3$ and thus violates EQX.
Moreover for any $j \geq 3$ and $j\notin N_1$, it must hold that $B_j\cap O^2=\emptyset$ as if $B_j\cap O^2 \neq \emptyset$, agent $j$ would receive two items from $O^2\cup O^3$ and hence violate EQX when compared to any agent $i\in N_1$; recall $c_i(B_i)=5T$.
Therefore, bundle $O^2$ must be entirely allocated between agents 1 and 2 in allocation $B$. 



We now prove that given all the above-mentioned restrictions, allocation $B$ can never be EQX. This is derived by the comparison among agents 1, 2 and agents in $N_1$.
Thereafter fix $t\in N_1$. By the construction of $N_1$, there exists $o\in B_t$ such that $c_t(o) = 0 $. Then, $\max_{o\in B_t} c_t(B_t\setminus\{o\}) = 5T$. Thus in order to make $B$ be EQX, $c_i(B_i)\geq 5T$ must be true for all $i\in [2]$. 
For any $i\in [2]$, if agent $i$ receives only $o^3_i$ from $O^3$, then due to ``no'' \ssc{partition} instance, there must be an agent $l\in [2]$ such that $c_l(B_l)<5T$, in which case allocation $B$ is not EQX.
Accordingly, at least one of agents $i\in [2]$ must receive some item $o^3_j$ with $j\neq i$. Without loss of generality, let agent 1 be such an agent. 
Then agent 1 only receives one item in allocation $B$ as $c_t(B_t) = 5T$ and every $o^3_j$ with $j\neq 1$ yields a cost of $10T$ for agent 1.
Accordingly, all $o^2_j$'s would be allocated to agent 2. Recall that agent 2 also receives one item from $O^3$, then we have $\max_{o\in B_2} c_2(B_2\setminus\{ o \}) >5T$, meaning that agent 2 violates EQX when compared to agent $t$. This yields the desired contradiction and thus allocation $B$ can never be EQX if no agent receives their big-item.
Therefore, there exists at least one agent receiving her big-item in allocation $B$. Let agent $i_1$ be one of such agents.

We next show that there exist at least two agents receiving their big-item in allocation $B$. For the sake of contradiction, suppose that agent $i_1$ is the only agent receiving their big-item. 
Then by the condition of EQX, it must hold that $|B_{i_1}| = 1$.
We can indeed remove agent $i_1$ and $o^1_{i_1}$ and only consider the reduced instance with $n-1$ agents and $2n+r-1$ items as other agents never violate EQX when comparing to agent $i_1$.
With a slight abuse of notations, let $N_1=\{i: B_i \cap O^1 \neq \emptyset \}$ and similarly $|N_1| \geq 2$. Then by arguments similar to the proof of Claim, 
agents 1 and 2 are not in $N_1$; this holds no matter whether agent 1 or 2 is in the reduced instance.
Then again, for any $i\in N_1$, $c_i(B_i)=5T$ holds and moreover $B_i\cap O^3 = \{ o^3_i\}$.
For the reduced instance, the total number of agents not in $N_1$ is $n-1-|N_1|$ as the reduced instance contains $n-1$ agents. Since each agent $i\in N_1$ receives exactly one item from $O^3$ in allocation $B$,
the total number of unallocated items from $O^3$ is $n-|N_1|$, and these $n-|N_1|$ items would be allocated to $n-1-|N_1| $ agents. Then there exists one agent receiving at least two items from $O^3$ and that agent violates EQX when compared to every agent $i\in N_1$, yielding the desired contradiction.
Therefore in allocation $B$, there exist at least two agents receiving their big-item.
We again eliminate two agents and their big-items and then consider the reduced instance with $n-2$ agents and $2r+n-2$ items.
By similar arguments, we can prove that there exist at least three agents receiving their big-items in allocation $B$. We repeat this argument and can show that all $n$ agents receive their big-items in allocation $B$.
Therefore, the social cost of EQX allocation $B$ is $nK-o(K)$.

By scaling the total cost of every agent to 1, we can claim that
it is NP-hard to approximate SC-EQX within an additive term of
$$
\frac{nK-o(K)}{K} \rightarrow n \textnormal{ as } K \rightarrow +\infty, 
$$
which yields the desired inapproximation factor in the statement of Theorem \ref{thm:hard-eqx-sc}.
\end{proof}

For the reduction in the proof of Theorem \ref{thm:hard-eqx-sc}, if the \ssc{partition} instance has an answer ``yes'', there exists an EQX allocation with social cost $5nT$, while if the answer is ``no'', every EQX allocation has social cost close to $nK$.
When $K$ tends to infinity, the reduction also indicates no polynomial time algorithm can achieve a bounded multiplicative approximation factor for SC-EQX!

\begin{corollary}
\label{cor:EQX:multiplicative}
     For any $n\geq 2$, there is no polynomial time algorithm having a bounded multiplicative approximation factor for SC-EQX unless $P = NP$.
\end{corollary}



\section{The Social Cost of EQ1 Allocations}
\label{sec:EQ1}

In contrast to the negative results on EQX allocations, in this section we prove that we can always find an EQ1 allocation with social cost no more than 1 in polynomial time. 
We then prove that for any $n\geq 3$, the problem of computing the optimal EQ1 allocations does not admit a polynomial time algorithm with additive loss smaller than $\frac{1}{2}$ unless P = NP.
When $n=2$, we give an additive approximation scheme.

\subsection{Bounding the Cost of EQ1}


To bound the cost of EQ1, we first design an algorithm that returns an EQ1 allocation whose social cost is at most 1. The algorithm starts with the optimal solution $S^*$
and implements a sequence of items reassignment.
Given an intermediate allocation $A$, we execute items reassignment based on directed graph $G_A=(N,E_A)$ constructed as follows: each agent $i\in N$ forms a vertex of $G_A$ and there is a directed edge or an arc $(i,j)\in E_A$ from $i$ to $j$ if $c_i(A_i\setminus\{ o\}) > c_j(A_j)$ for all $o\in A_i$.
In other words, an arc from $i$ to $j$ indicates that agent $i$ violates EQ1 when compared to agent $j$. We classify agents into three categories based on their in- and out-degrees:
\begin{align}\label{eq:GA:N}
\begin{split}
    & N_0 \text{ contains all agents with zero in-degree and zero out-degree in $G_A$;}\\
    & N_1 \text{ contains all agents with zero in-degree and positive out-degree in $G_A$;}\\
    & N_2 \text{ contains all agents with positive in-degree in $G_A$.}
\end{split}
\end{align}
The algorithm each time identifies an agent in $N_1$ and assigns one item from her bundle to the agent with the minimum cost in $N_2$. After each reassignment, we update $G_A$, $N_0$, $N_1$, $N_2$ and repeat the process until $N_1$ becomes empty. 
The formal description of the algorithm is shown in  Algorithm \ref{alg:eq1}.

\begin{algorithm}[h]
    \caption{Computing EQ1 allocations with social cost $\le$ 1}\label{alg:eq1}
\begin{algorithmic}[1]
\REQUIRE{An instance $\cI=\langle N,O,\{c_i\}_{i\in N} \rangle$}
\ENSURE Allocation $A=(A_1,\ldots,A_n)$.
\STATE Compute the optimal allocation $S^*=(S_1^*,\ldots, S^*_n)$ by allocating each item to \\
the agent with the smallest cost, where ties are broken arbitrarily.
\STATE Initialize $A_i\leftarrow S^*_i$ for all $i \in N$ and construct $G_A=(N,E_A)$ where there is a directed edge $(i,j)$ for $i,j \in N$ if $c_i(A_i\setminus\{ o\}) > c_j(A_j)$ for all $o\in A_i$. Define \\
$N_0$, $N_1$ and $N_2$ as Equation (\ref{eq:GA:N}).
\WHILE{$N_1\neq \emptyset$}
    \STATE Let $i^* \in \arg\max_{i\in N_1}(\min_{o\in A_i} c_i(A_i \setminus \{o\}))$ and arbitrarily pick $o' \in A_{i^*}$.\label{alg1::step:choose-largest-cost-agent}
    \STATE Let $j^* \in \arg\min_{j \in N_2}c_i(o') $ and make the reassignment $A_{i^*} \gets A_{i^*}\setminus \{o' \}$ 
    \\and $A_{j^*}\gets A_{j^*}\cup\{o'\}$.\label{alg1::step:choose-agent-least-cost}
    \STATE Update $G_A = (N,E_A)$, $N_0$, $N_1$ and $N_2$.
\ENDWHILE
\end{algorithmic}
\end{algorithm}

We first present several lemmas and terminologies used in the below proof. 
Each iteration of the while loop in Algorithm \ref{alg:eq1} is called a round. 
We focus on two time frames of a given round $t$: \emph{at the beginning of round }$t$ refers to the moment right before the item reassignment happens; and \emph{after round} $t$ refers to the moment right after the item movement happens and the graph and corresponding sets get updated.

\begin{observation}\label{obs:eq1}
    Throughout the while loop of Algorithm \ref{alg:eq1}, the minimum cost of an agent is non-decreasing.
\end{observation}
To see Observation \ref{obs:eq1}, let us compare agents' cost before and after some round $t$. There exists exactly one agent (let's say agent $i^*$) whose cost is weakly decreased 
and moreover this agent belongs to $N_1$ at the beginning of round $t$. After round $t$, the cost of agent $i^*$ should be strictly larger than the minimum cost of an agent in the allocation at the beginning of round $t$; otherwise right before round $t$, agent $i^*$ satisfies EQ1 and should not be in $N_1$.


\begin{lemma}\label{lem::eq1-1}
    For agent $i$, if $i \in N_0$ at the beginning of some round, then $i\in N_0$ holds in all future rounds.
\end{lemma}
\begin{proof}
    For the sake of contradiction, assume that agent $i$ has zero in- and out-degree at the beginning of round $t$ while $i$ has positive in- or out-degree after round $t$.
    Throughout the proof, let $A'$, $G_{A'}$, $N'_0$, $N'_1$, $N'_2$ refer to the allocation, the graph and the corresponding sets of agents after round $t$;
    let $A^b$, $G_{A^b}$, $N^b_0$, $N^b_1$ and $N^b_2$ refer to the allocation, the graph and the corresponding sets of agents at the beginning of round $t$. By the construction, $i\in N^b_0$.

    If $i\in N'_1$, by Observation \ref{obs:eq1}, it holds that $\min_{j}c_j(A'_j) \geq \min_{j} c_j(A^b_j)$. As $i\in N^b_0$, 
    agent $i$ does not violate EQ1 in allocation $A_i'$ and moreover $A^b_i = A'_i$.
    Combining these facts, the out-degree of $i$ in $G_{A'}$ should be 0 and thus $i\notin N'_1$, deriving the desired contradiction.
    

    If $i\in N'_2$, suppose that in $G_{A'}$ there is an arc $(j,i)$. 
    According to the reallocation rule, the appearance of arc $(j,i)$ can only be caused by the fact that agent $j$ takes some item $o'$ in round $t$. Hence, $j\in N_2^b$ holds. Then we claim that $c_j(A^b_j) < c_i(A^b_i)$; otherwise, the vertex pointing to $j$ in $G_{A^b}$ should also point to $i$ in $G_{A^b}$, making $i\in N^b_2$, a contradiction.
    For allocation $A'$, we have $$\min_{o\in A'_j }c_j(A'_j\setminus\{o\}) \leq c_j(A'_j\setminus\{o'\}) = c_j(A_j^b) < c_i(A^b_i) = c_i(A'_i),$$
    which indicates that arc $(j,i)$ does not exist in $G_{A'}$, deriving the desired contradiction.
\end{proof}

\begin{lemma}\label{lem::eq1-2}
    For agent $i$, if in $G_{S^*}$, the graph  corresponding to the optimal allocation $S^*$, the in-degree of vertex $i$ is 0, then the vertex $i$ has in-degree 0 in $G_A$ for every intermediate allocation $A$.
\end{lemma}
\begin{proof}
    For the sake of contradiction, suppose that some vertex $i$ has in-degree 0 in $G_{S^*}$ but has positive in-degree after some round $t$. 
    Suppose that right after round $t$ is the earliest moment such that vertex $i$ has a positive in-degree.
    Throughout the proof, let $A$, $G_A$, $N_0$, $N_1$, $N_2$ refer to the allocation, the graph and the corresponding sets of agents after round $t$;
    let $A^b$, $G_{A^b}$, $N^b_0$, $N^b_1$, $N^b_2$ refer to the allocation, the graph and the corresponding sets of agents at the beginning of round $t$. Then the contradiction assumption implies that $i\notin N^b_2$ and $i\in N_2$. In the following, we prove no such an agent $i$ exists.
    
    As $i\notin N^b_2$, there are two possibilities: (i) $i\in N^b_1$, and (ii) $i\in N^b_0$. According to Lemma \ref{lem::eq1-1}, it must hold that $i\in N^b_1$. We below split the proof by discussing whether an item is taken from the bundle of agent $i$ in round $t$.

    If an item is taken from the bundle of agent $i$ in round $t$, i.e., $|A^b_i\setminus A_i| = 1$,
    then agent $i$ is chosen by the Algorithm \ref{alg:eq1} at Step \ref{alg1::step:choose-largest-cost-agent} in round $t$. Then for any $j\in N^b_1$ and any $o'\in A_j^b$, we have $c_j(A^b_j\setminus\{o'\}) \leq \min_{o\in A_i^b} c_i(A_i^b\setminus\{o\}) \leq c_i(A_i)$.
    Since for any $j\in N^b_1\setminus\{i\}$, $A_j^b=A_j$ holds, then there is no arc $(j,i)$ in graph $G_A$ as every agent $j$ satisfies EQ1 when compared to agent $i$ in $A$. By Observation \ref{obs:eq1}, there is no arc pointing from some vertex in $N^b_0$ to $i$.
    To show $i\in N_2$ never holds, the remaining is to prove that there is no arc $(j,i)$ with $j\in N^b_2$ in $G_A$. Let us first prove the following claim.
    \begin{claim}\label{claim::ak1}
    For any $j\in N^b_2$, arc $(i,j)$ exists in $G_{A^b}$.
    \end{claim}
    \begin{proof}[Proof of Claim \ref{claim::ak1}]
        We first present properties satisfied by every intermediate graph $G_{A'}$. Let us now consider a fixed graph $G_{A'}$ and corresponding sets $\{N'_i\}$. There is no cycle in $G_{A'}$ as any arc $(x,y)$ indicates $c_x(A'_x) > c_y(A'_y)$ and the existence of a cycle (containing some vertex $x$) implies $c_x(A'_x)>c_x(A'_x)$, a contradiction. 
        Then consider any three vertices $x,y,z$. If there are arcs $(x,y)$ and $(y,z)$, then there exists an arc $(x,z)$ as for any $o\in A'_x$, $c_x(A'_x\setminus\{o\})>c_y(A'_y)>c_z(A'_z)$. 
        Accordingly, for any $q\in N'_2$, there exists in $G_{A'}$ an arc $(p,q)$ pointing from some $p\in N'_1$ to $q$. The reason is that vertex $q$ is on some path as there is no cycle, and the starting vertex $s$ of that path has in-degree 0 and positive out-degree and hence $s\in N'_1$.
        By the above argument, arc $(s,q)$ exists in $G_{A'}$.

        Now we are ready to prove Claim \ref{claim::ak1}. By above arguments, for any $j\in N^b_2$, there exists an arc $(j', j)$ with $j'\in N^b_1$. By the construction of $G_{A^b}$, arc $(j', j)$ implies $c_{j'}(A^b_{j'}\setminus\{o\})>c_{j}(A^b_{j})$ for all $o\in A^b_{j'}$.
        According to Step \ref{alg1::step:choose-largest-cost-agent}, agent $i$ satisfies $\min_{o\in A^b_i}c_i(A^b_i\setminus\{o\})\geq \min_{o\in A^b_{q}}c_{q}(A^b_{q}\setminus\{o\})$ for all $q\in N^b_1$, and hence for any $j\in N^b_2$, $\min_{o\in A^b_i\setminus\{o\}}c_i(A^b_i)>c_j(A^b_j)$ holds.
        Therefore for any $j\in N^b_2$, arc $(i,j)$ exists in $G_{A^b}$.
    \end{proof}
     \noindent For any $j\in N_2^b$, by Claim \ref{claim::ak1}, it is not hard to verify that she does not violate the condition of EQ1 in allocation $A$ when compared to agent $i$. Hence, arc $(j,i)$ does not exist in $G_A$.

     We now consider the case where agent $i$ is not chosen at Step \ref{alg1::step:choose-largest-cost-agent} in round $t$, i.e., $A^b_i=A_i$. Let $l$ be the agent receiving an item in round $t$. For any $j\neq l$, no arc $(j, i)$ exists in $G_A$ as $A_j=A_j^b$.
     As for agent $l$, it is clear that $l\in N^2_b$ and thus $c_i(A^b_i) > c_l(A^b_l)$ as $i\in N^b_1$; note if $c_i(A^b_i) \leq c_l(A^b_l)$, vertex $i$ would have a positive in-degree in $G_{A^b}$, a contradiction. 
     The inequality $c_i(A^b_i)>c_l(A^b_l)$ implies that agent $l$ does not violate EQ1 in $A$ when compared to agent $i$ as $A_i=A^b_i$ and $|A_l\setminus A^b_l| = 1$. Therefore, there is no arc $(l,i)$ in $G_A$. 
\end{proof}

\begin{theorem}\label{thm::eq1-sc-at-most-1}
    Algorithm \ref{alg:eq1} runs in polynomial time and returns an EQ1 allocation with a social cost of at most 1.
\end{theorem}
\begin{proof}
    By the graph construction, $N_1=\emptyset$ implies $N_2=\emptyset$ and vice versa.
    As after each round, the cardinality of the union bundle of agents in $N_2$ is increased by 1, then the graph $G_A$ varies once in at most every $m$ rounds.
    Based on Lemmas \ref{lem::eq1-1} and \ref{lem::eq1-2}, 
    each time when $G_A$ varies, we have three possible situations: (i) one agent moves from $N_1$ to $N_0$, (ii) one agent moves from $N_2$ to $N_1$ or to $N_0$, and (iii) both (i) and (ii) happen. As a consequence, $N_1$ and $N_2$ will be $\emptyset$ after $G_A$ getting updated (at most) $2n$ times. Thus, Algorithm \ref{alg:eq1} terminates in $2nm$ rounds.
    The running time of executing a round is at most $O(mn)$ and therefore Algorithm \ref{alg:eq1} terminates in $O(m^2n^2)$ time. 
    
    
    The fact that returned allocation $A$ is EQ1 directly follows from the construction of $G_A$ and $N_1, N_2=\emptyset$ at termination. We now bound the social cost of $A$. Let $T$ be the last round of the algorithm, i.e., after round $T$, $N_1, N_2=\emptyset$.
    Suppose $N'_2$ is the set of vertex with positive in-degree at the beginning of round $T$ and suppose $N'_2 \ni k$.
    Denote by $Q$ the set of items allocated to different agents in $S^*$ and $A$.
    For any $o\in Q$, suppose $o\in A_j$, and then we have $c_j(o) \leq c_k(o)$ by Step \ref{alg1::step:choose-agent-least-cost} and the fact that $k$ is in $N_2$ in every round $t=1,...,T$ (Lemmas \ref{alg:eq1} and \ref{lem::eq1-2}). Accordingly, we have the following inequality,
    $$
    \SC(A) \leq \sum_{o\in Q} c_k(o) + \sum_{i\in N} c_i(S^*_i\setminus Q) \leq \sum_{o\in Q} c_k(o) + \sum_{i\in N} c_k(S^*_i\setminus Q) = c_k(O)=1,
    $$
    where the second inequality transition is because $S^*$ is the optimal allocation.
\end{proof}

\begin{theorem}\label{thm::cof-eq1}
    The cost of EQ1 is at most 1 and at least $1-\frac{1}{n}$ for any $n\geq 2$.
\end{theorem}
\begin{proof}
    The upper bound follows directly from Theorem \ref{thm::eq1-sc-at-most-1}. As for the lower bound, let us consider the following instance. There are a set $\{1,\ldots,n\}$ of $n$ agents and a set $O=\{o_1,\ldots,o_n,o_{n+1}\}$ of $n+1$ items. Let $\epsilon>0$ be an arbitrarily small number. The cost functions $\{c_i\}_{i\in [n]}$ of agents are described as follows;
    \begin{itemize}
        \item for agent 1: $c_1(o_j) = \epsilon$ for all $j\leq n$ and $c_1(o_{n+1}) = 1-n\epsilon$;
        \item for agent $i\geq 2$: $c_i(o_j)=\frac{1}{n}$ for all $j\leq n$ and $c_i(o_{n+1}) = 0$.
    \end{itemize}
    In an optimal allocation $S^*$, agent 1 receives $S^*_1=\{o_1,\ldots,o_n\}$ and agent 2 receives $o_{n+1}$, resulting in $\SC(S^*) = n\epsilon$. However, agent 1 violates EQ1 in $S^*$.
    Let us consider allocation $A=(A_1,\ldots,A_n)$ with $A_i=\{o_i\}$ for $i\in [n-1]$ and $A_n=\{o_n,o_{n+1}\}$. It is not hard to verify that $A$ is an EQ1 allocation and has social cost $\frac{n-1}{n}+\epsilon$.
    Moreover, no other EQ1 allocation has social cost less than $\SC(A)$; any EQ1 allocation better than $A$ must not allocate $o_{n+1}$ to agent 1 and then in such an allocation, agent 1 receives at most one item.
    Therefore, the cost of fairness regarding EQ1 is at least $\frac{n-1}{n}-(n-1)\epsilon$ tending to $1-\frac{1}{n}$ when $\epsilon \rightarrow 0$.
\end{proof}

\subsection{Computing the Optimal EQ1 Allocations}
We now consider the optimization problem of minimizing social cost subject to EQ1 constraint, denoted by SC-EQ1. 
It is known that SC-EQ1 is NP-hard \cite{DBLP:journals/aamas/SunCD23}, where the reduction, however, does not suggest any hardness of approximation.
We below show that SC-EQ1 is 
hard to approximate within some constant even when there are only three agents.

\begin{theorem}\label{thm:hard-eq1-sc}
    For any $n\geq 3$, 
    there is no polynomial time algorithm approximating SC-EQ1 within an additive factor smaller than $\frac{1}{2}$ unless P = NP.
\end{theorem}
\begin{proof}
      We present the reduction from the \ssc{Partition} problem. Given an arbitrary instance of \ssc{partition}, we construct a fair division instance $\cI$ with $n$ agents and a set $\{ o_1,\ldots, o_{r+n-1} \}$ of $r+n-1$ items.
      The cost functions of agents are shown in Table \ref{tab:eq1-hard-general} where $K$ is a sufficiently large number, i.e., $K\gg T$. 
      For any $i\in [n]$, we have $\sum_{j=1}^{r+n-1} c_i(o_j) = (2r+\frac{n-1}{K})T$ and thus the cost functions are normalized.

      \begin{table}[htbp]
  \centering
  \caption{The constructed fair division instance for Theorem \ref{thm:hard-eq1-sc}}
  \label{tab:eq1-hard-general}
    \begin{tabular}{|c|c|cccccccccc|}
    \hline
    \multicolumn{2}{|c|}{\multirow{2}[4]{*}{$c_i(o)$}} & \multicolumn{10}{c|}{Items} 
    \\
\cline{3-12}    \multicolumn{2}{|c|}{} & $o_1$ & $ o _2$ & $\cdots$ & $o_r$ & $o_{r+1}$ & $o_{r+2}$ & $o_{r+3}$ & $o_{r+4}$ & $\cdots$ & $o_{r+n-1}$ 
\\
    \hline
    \multirow{5}[2]{*}{Agents} & $1$ & $\frac{p_1}{K}$ & $ \frac{p_2}{K}$ & $\cdots$ & $ \frac{p_r}{K}$ & $ r T$ & $r T$ & $\frac{T}{K}$ & $\frac{T}{K}$&  $\cdots$  & $\frac{T}{K}$ \bigstrut[t]\\
          & $2$   & $\frac{p_1}{K}$ & $ \frac{p_2}{K}$ & $\cdots$ & $ \frac{p_r}{K}$ & $ r T$ & $r T$  & $\frac{T}{K}$ & $\frac{T}{K}$&  $\cdots$  & $\frac{T}{K}$\\
        & $\vdots$ & $\vdots$ & $\vdots$& $\ddots$& $\vdots$& $\vdots$& $\vdots$& $\vdots$& $\vdots$ & $\ddots$& $\vdots$ \\
        & $n-1$   & $\frac{p_1}{K}$ & $ \frac{p_2}{K}$ & $\cdots$ & $ \frac{p_r}{K}$ & $ r T$ & $r T$  & $\frac{T}{K}$ & $\frac{T}{K}$&  $\cdots$  & $\frac{T}{K}$\\
          & $n$ & $2T$ & $2T$ & $\cdots$ & $2T$ & $\frac{T}{K}$ & $\frac{T}{K}$ & $\frac{T}{K}$ & $\frac{T}{K}$&  $\cdots$  & $\frac{T}{K}$  \\
    \hline
    \end{tabular}%
\end{table}%
\noindent 

Suppose we have a ``yes'' instance of \ssc{partition} and let $I_1$ and $I_2$ be a solution. We now consider allocation $A = (A_1,\ldots, A_n)$ where $A_1 = \bigcup_{j\in I_1} o_j$, $A_2 = \bigcup_{j\in I_2} o_j$, $A_n = \{ o_{r+1}, o_{r+2} \}$ and for any $3\leq j \leq n-1$, bundle $A_j$ contains exactly one item from $\{o_{r+3},\ldots,o_{r+n-1}\}$. 
It is not hard to verify that allocation $A$ is EQ1 and achieves social cost $\SC(A) = \frac{(n+1)T}{K}$.
Moreover, $A$ is the optimal EQ1 allocation.

For the other direction, suppose we have a ``no'' instance of \ssc{partition}. We first claim that in any EQ1 allocation $B$, the bundle of agent $n$ does not include both $o_{r+1}$ and $o_{r+2}$. 
Due to ``no'' instance, assigning $\{ o_1,\ldots,o_r, o_{r+3},\ldots,o_{r+n-1} \}$ to the first $n-1$ agents makes one of them receive a cost less than $\frac{T}{K}$. 
If agent $n$ receives $\{ o_{r+1}, o_{r+2} \}$, then she violates the property of EQ1 when compared to the agent with the least cost.
Thus in allocation $B$, at least one of $o_{r+1}, o_{r+2}$ are assigned to first $n-1$ agents in allocation $B$. Accordingly for an EQ1 allocation $B$, its social cost $\SC(B) \geq (r + \frac{n}{K})T$.

By scaling the total cost of every agent to 1, we can claim that
it is NP-hard to approximate SC-EQ1 within an additive factor of
$$
\frac{r - \frac{1}{K}}{2r+\frac{n-1}{K}} \rightarrow \frac{1}{2}\textnormal{ as } K\rightarrow +\infty,
$$
which yields the desired inapproximation factor in the statement of Theorem \ref{thm:hard-eq1-sc}.
\end{proof}

Recall that by Theorem \ref{thm::eq1-sc-at-most-1}, Algorithm \ref{alg:eq1} can return an EQ1 allocation with social cost at most 1 in polynomial time.
Thus, Algorithm \ref{alg:eq1} can approximate SC-EQ1 within an additive factor of 1, which complements the above negative result.

We remark that the reduction in the proof of Theorem \ref{thm:hard-eq1-sc}
also implies that approximating SC-EQ1 within a multiplicative factor of $\frac{rK+n}{n+1}$ is NP-hard. 
As $K$ can be arbitrarily large, we can then claim that no polynomial time algorithm can have bounded multiplicative approximation guarantee.

\begin{corollary}
    For $n\geq 3$, there is no polynomial-time algorithm that has a bounded multiplicative approximation factor for SC-EQ1 unless $P = NP$.
\end{corollary}

As for the case of two agents, it is also known that SC-EQ1 is NP-hard \cite{DBLP:journals/aamas/SunCD23}. We below present an additive approximation scheme for SC-EQ1 when $n=2$. 
For the choice of parameter $h$, we follow the work of \cite{DBLP:conf/wine/CaragiannisI21} and set $h$ as the total cost of an agent for all items, i.e., $h=1$.
Denote by $\OPT^{EQ1}$ the social cost of the optimal EQ1 allocation.

\begin{theorem}\label{thm:eq1-ptas}
    Let $\epsilon>0$ be an accuracy parameter. Given an instance of SC-EQ1 with two agents and $m$ items,
    the algorithm runs in time $O(m2^{\frac{2}{\epsilon}})$ and computes an EQ1 allocation with a social cost of at most $\OPT^{EQ1} + \epsilon$.
\end{theorem}

Our algorithm makes use of the PTAS idea for the classical scheduling problem of makespan minimization on identical machines \cite{1146b87e-3eb4-3402-aa54-16747a111282}.
We first partition the space of complete allocations $\mathcal{D}$ into $q$ districts $\mathcal{D}^{(1)}, \mathcal{D}^{(2)},\ldots, \mathcal{D}^{(q)}$ such that $\mathcal{D}^{(1)} \cup \cdots \cup \mathcal{D}^{(q)} = \mathcal{D}$
where $q$ is polynomial in $m$ and $\frac{1}{\epsilon}$. 
Then for each district $\mathcal{D}^{(l)}$, we compute
a representative allocation that is EQ1 and has a social cost close to the optimal objective value in $\mathcal{D}^{(l)}$.
At last, we pick the best representative among all districts.

We now present the algorithm in detail. Based on the accuracy parameter $\epsilon>0$, we classify items into two categories as follows: for any $o\in O$,
\begin{align*}
    \text{it is a \emph{big item} if } \exists i =1,2 \text{ such that } c_i(o)>\epsilon; \text{it is a \emph{small item} otherwise.}
\end{align*}
As the total cost of every agent is 1, the total number of big items is at most $\frac{2}{\epsilon}$. We define districts $\{D^{(l)}[i,o]\}$ as follows: two allocations $F^1, F^2$ lie in district $D^{(l)}[i,o]$ if $F^1$ allocates every big item to the same agent as $F^2$ and moreover item $o$ is allocated to $i$ in both $F^1$ and $F^2$.
Note that the requirement of the identical assignment of big items in $F^1, F^2$ does not violate the requirement of allocating item $o$ to agent $i$ no matter whether $o$ is a big item or a small item.
There are $2^{\frac{2}{\epsilon}}$ possible allocations of big items and at most $m$ options for the choice of the specific $o$ and two options for agent $i$.
Thus, the total number of districts is at most $2m2^{\frac{2}{\epsilon}}$.

We next describe how to find the representative within each district. Consider a fixed district $D^{(l)}[1,o']$ (the district for agent 2 is symmetric). For every allocation in this district, assignments of big items and item $o'$ have already been fixed while that of small items (possibly excepting $o'$) may vary.
Let $s_1,\ldots,s_k$ be small items;
note $k<m$. Let $a_1$ be the cost of agent 1 with respect to the fixed big items, excluding $o'$ if $o'$ is a big item, allocated to her 
and let $a_2$ be the cost of agent 2 with respect to the fixed big items allocated to her.
We determine the assignment of small items based on the optimal solution of the following linear programming. In the LP, variable $x_j$ refers to the fraction of $s_j$ allocated to agent 1. The objective function is the total cost of small items and
the first two constraints guarantee that (i) agent 1 receives less cost after removing $o'$ and (ii) agent 2 receives less cost than agent 1.
Note that we can get a complete (possibly fractional) allocation by extending the fixed partial assignment of big items and $o'$ according to the solution of the LP.
Moreover, the assignment of small items in those integral EQ1 allocations in district $\mathcal{D}^{(l)}[1,o']$, where $o'$ is the largest cost item in the bundle of agent 1 and the cost of agent 1 is at least that of agent 2, corresponds to feasible solutions of the LP.

$$
\centering
\begin{array}{ll}\min & \sum_{j=1}^k c_1(s_j)x_j + \sum_{j=1}^kc_2(s_j)(1-x_j)\\ \text { s.t. } & a_1+ \sum_{j=1}^k c_1(s_j)x_j \leq a_2 + \sum_{j=1}^kc_2(s_j)(1-x_j) \\ 
& a_2 + \sum_{j=1}^kc_2(s_j)(1-x_j) \leq  a_1 + c_1(o')+ \sum_{j=1}^k c_1(s_j)x_j \\ 
& 0\leq x_j \leq 1, \quad j=1, \ldots, k .
\end{array}
$$

We can compute in time polynomial to $m$ the optimal solution $\{x_j^*\}$ of the LP. If $c_1(o')=0$, the first two inequality constraint becomes an equality constraint. 
When $c_1(o')\neq 0$, the first two inequality constraints can not be fulfilled with equality at the same time.
As a consequence, since the solution $\{x_j^*\}$ is a basic feasible solution of the underlying polyhedron in $k$-dimensional space, at least $k-1$ values $x_j^*$ should be integral, i.e., at most one $x_j^*$ is fractional. 
To find the representative in $D^{(l)}[1,o']$, we allocate small items as follows: item $s_j$ is allocated to agent 1 if $x_j^*=1$ and is allocated to agent 2 if $x_j^*=0$.
If every $x^*_j$ is integral,
then we find the representative. 
If $\{x_j^*\}$ is not an integral solution, as we have already argued, there is at most one value $x_t^*$ being fractional.
We next describe how to allocate item $s_t$.
Define $E_1=\{s_j: j\in [k] \textnormal{ and } x_j^*=1\}$ and $E_2=\bigcup_{j\in[k]}s_j \setminus (E_1\cup s_t)$, that is, $E_1$ (resp. $E_2$) refers to the set of small items determined to be allocated to agent 1 (resp. agent 2) based on $\{x^*_j\}$. Up to here, the assignment of every big item, item $o'$ and small items excluding $s_t$ has been determined; recall that $o'$ is allocated to agent 1.

If $a_1+c_1(E_1) \geq a_2+c_2(E_2)$, then item $s_t$ is allocated to agent 2. {One can verify that this complete allocation is EQ1.}
We now prove that this complete allocation is EQ1. By the first constraint, we have $a_1+c_1(E_1)+c_1(s_t)x_t^* \leq a_2+c_2(E_2)+c_2(s_t)(1-x_t^*)$, implying $a_1+c_1(E_1) \leq a_2+c_2(E_2) + c_2(s_t)$.
Thus in the final allocation, agent 1 is EQ1 as after removing item $o'$, her cost is no greater than that of agent 2. 
As for agent 2, by the if-condition of $a_1+c_1(E_1)\geq a_2+c_2(E_2)$, her cost would be at most that of agent 1 after removing $s_t$. Therefore, the final allocation is EQ1.

If $a_1+c_1(E_1)<a_2+c_2(E_2)$, we compare $a_1+c_1(E_1)+c_1(o')$ and $a_2+c_2(E_2)$ and then allocate $s_t$ to the agent with a smaller cost; if $a_1+c_1(E_1)+c_1(o')>a_2+c_2(E_2)$, we allocate $s_t$ to agent 2 and otherwise allocate $s_t$ to agent 1.
The agent receiving $s_t$ satisfies EQ1 in the final allocation. Now we consider the agent not receiving $s_t$.
If agent 2 does not receive $s_t$, the second constraint of LP implies that agent 2 satisfies EQ1. 
If agent 1 does not receive $s_t$, the first constraint of LP indicates that after removing $o'$, the cost of agent 1 is less than that of agent 2 in the final allocation. Therefore, the final allocation is EQ1.

We now have found the representative for each district and have shown that the representative allocation is EQ1.
Suppose $B^*=(B^*_1,B^*_2)$ is the optimal EQ1 allocation 
and without loss of generality $c_1(B_1)\geq c_2(B_2)$. Moreover, item $o^* \in B_1$ satisfies $c_1(o^*) \geq c_1(o)$ for all $o\in B_1$.
Then by the construction of districts, allocation $B^*$ must lie in some district $\mathcal{D}^{(l^*)}[1,o^*]$. Let $A^{l^*}$ be the representative of $\mathcal{D}^{(l^*)}[1,o^*]$ and $z^{l^*}$ be the optimal objective value of the LP corresponding to $\mathcal{D}^{(l^*)}[1,o^*]$.
Then we have
$$
\SC(A^{l^*}) \leq a_1+a_2+z^{l^*} + \epsilon \leq \SC(B^*) + \epsilon.
$$
The first inequality transition is because $x_t^*$ is the only fractional value in $\{x_j^*\}$ and $s_t$ is a small item. 
The second inequality transition is because the assignment of small items (possibly excluding $o'$) of $B$ corresponds to a feasible solution of LP.
After visiting all the representatives, we pick the allocation with the minimum social cost. It is clear that the social cost of that allocation is no greater than $\SC(A^{l^*})$.
Therefore, the algorithm runs in time $O(m2^\frac{2}{\epsilon})$ and computes an EQ1 allocation with a social cost of at most $\OPT^{EQ1} +\epsilon$. 




\section{The Social Cost of EF1 Allocations}

\label{sec:EF1}

Our last fairness criterion is EF1. As we will show the results are positive and similar as EQ1 allocations. 

\subsection{Bounding the Cost of  EF1}

\begin{theorem}\label{thm::ef1-sc-at-most-1}
    The cost of EF1 is at most $1$ and at least $1-\frac{1}{n}$.
\end{theorem}
    

The high-level idea of the proof of Theorem \ref{thm::ef1-sc-at-most-1} is that we implement the round-robin algorithm on $n$ different orderings of the agents. 
It is clear that all these $n$ algorithms return EF1 allocations, and we further prove that one of them returns an allocation with a social cost of at most 1.
    \begin{proof}[Proof of Theorem \ref{thm::ef1-sc-at-most-1}]
        We first consider the upper bound and show that there always exist EF1 allocations with social cost at most 1.
        Given an agent $i$ with cost function $c_i$, suppose $c_i(o_{i_1})\le c_i(o_{i_2}) \le \cdots \le c_i(o_{i_m})$.
        We partition $O$ into $n$ bundles in a round-robin way, i.e., $M^i_k = \{o_i, o_{n+i}, o_{2n+i}, \ldots\}$ for $k = 1, \ldots, n$.
        It is clear that 
        \begin{equation}\label{eq:ef1:cof:1}
            c_i(M^i_1) + c_i(M^i_2) + \cdots + c_i(M^i_n) = 1, \text{ for any $i\in N$.}
        \end{equation}
        Moreover, in a round-robin algorithm where agent $i$ is the $k$-th to select items, her cost is no greater than $c_i(M^i_k)$. This is because in any round $l\ge 1$ of the algorithm, $l\cdot n + i - 1$ items have been allocated at the moment when it is $i$'s turn to select an item and thus this item's cost cannot be greater than $c_i(o_{i_{l\cdot n + i}})$.
        
        Next, consider $n$ permutations of the agents: 
        \begin{table}[H]
            \centering
            \begin{tabular}{ccccccc}
                $\sigma_1$ & = & $(1$, & $2$, &$\ldots$ &$n-1$,& $n)$\\
                $\sigma_2$ & = & $(2$, & $3$, &$\ldots$ &$n$, &$1)$ \\
                &\vdots\\
                $\sigma_n$ & = & $(n$, & $1$, &$\ldots$ &$n-2$, & $n-1)$.
            \end{tabular}
        \end{table}
        \noindent Summing up the $n$ equalities in Equation (\ref{eq:ef1:cof:1}) and rearranging the terms according to the permutations $\sigma_1, \ldots, \sigma_n$ gives 
        \[
        \sum_{k=1}^n \sum_{l=1}^n c_{\sigma_k(l)}(M^{\sigma_k(l)}_{l}) = n.
        \]
        Thus
        \[
        \min_{\sigma_1, \ldots, \sigma_n} \sum_{i=1}^n c_{\sigma_k(l)}(M^{\sigma_k(l)}_{l}) \le 1.
        \]
        Let $A^{k}=(A^{k}_1,\ldots, A^{k}_n)$ be the allocation returned by the round-robin algorithm where the agents are ordered by $\sigma_k$.
        As we have discussed, $c_{\sigma_k(l)}(A^{k}_{l}) \le c_{\sigma_k(l)}(M^{\sigma_k(l)}_{l})$ for $l=1,\ldots, n$.
        Therefore,
        \[
        \min_{\sigma_1, \ldots, \sigma_n} \sum_{i=1}^n c_{\sigma_k(l)}(A^{k}_{l}) \le 1.
        \]
        That is, the allocation with the minimum cost among the $n$ allocations returned by the $n$ round-robin algorithms has social cost no greater than 1. 
        Moreover, such an allocation can be found in polynomial time.

        As for the lower bound, let us consider an instance with $n$ agents and a set $\{o_1,\ldots,o_{K+1}\}$ of $K+1$ items, where $K$ is a large number and divides $n$.
        The cost functions $c_i$'s are as follows;
        \begin{itemize}
            \item for agent 1: $c_1(o_j)=\epsilon$ for all $j\leq K$ and $c_1(o_{K+1})=1-K\epsilon$;
            \item for agent $i\geq 2$: $c_i(o_j)=\frac{1}{K}$ for all $j\leq K$ and $c_i(o_{K+1}) = 0$;
        \end{itemize}
        where $0<\epsilon\ll \frac{1}{K}$.

        In an optimal allocation $S^*=(S^*_1,\ldots,S^*_n)$, agent 1 receives $S^*_1=\{o_1,\ldots,o_K\}$ and agent 2 receives $o_{K+1}$, resulting in $\SC(S^*)=K\epsilon$. However, agent 1 violates EF1 in $S^*$ when compared to agent $j\geq 3$. In the following, we will show that the social cost of the optimal EF1 allocation $B=(B_1,\ldots,B_n)$ is close to $1-\frac{1}{n}$.
        It is not hard to verify that in the optimal EF1 allocation, $o_{K+1}$ is not allocated to agent 1.
        As the cost functions of agents $2,\ldots,n$ are identical, we can assume that $o_{K+1}$ is allocated to agent 2 in $B$ due to symmetry.
        Then for every $i\in[n]$, let $x_i$ refers to the number of the first $K$ items allocated to agent $i$.
        Since $\epsilon \ll \frac{1}{K}$, the upper bound of $x_1$ implies the lower bound of the social cost of $B$.
        The EF1 condition among agents $2,\ldots,n$ requires that for any $i,j=2,\ldots,n$, $|x_i-x_j| \leq 1$;
        EF1 condition among agents $1,3,\ldots,n$ requires that for any $j=3,\ldots,n$, $|x_1-x_j| \leq 1$;
        EF1 condition between agents 1 and 2 requires $x_2-x_1\leq 1$.
        Given these constraints, we have $x_1\leq \frac{K+n}{n}$ as $\sum_{j=1}^n x_j=K$.
        Then we have $\sum_{j=2}^nx_j \geq \frac{(n-1)K-n}{n}$ and thus
        $\SC(B)\geq \frac{n-1}{n}-\frac{1}{K}$.
        Therefore,
        the cost of fairness regarding to EF1 is at least $\frac{n-1}{n}-\frac{1}{K} - K\epsilon$,
        close to $1-\frac{1}{n}$ when $\epsilon \ll \frac{1}{K}\ll \frac{1}{n}$.
 \end{proof}

\subsection{Computing the Optimal EF1 Allocations}
Denote by SC-EF1 the optimization problem of minimizing social cost subject to EF1 constraint. We below first show that SC-EF1 is hard to approximate within a constant when $n\geq 3$. Then for the case of two agents, we present an additive approximation scheme.


\begin{theorem}\label{thm:hard-ef1-sc}
     For any $n\geq 3$, there is no polynomial time algorithm approximating SC-EF1 within an additive factor smaller than $\frac{n-1}{2n+2}$ unless P = NP.
\end{theorem}
\begin{proof}
    We will reduce from the problem of \emph{Partition-into-$k$-sets}: given a fixed number $k$ and a set $\{ p_1,\ldots, p_q \}$ of $q$ positive integers whose sum is $kT$, can $[q]$ be partitioned into $k$ indices sets $I_1,\ldots,I_k$ such that for every $j\in[k]$, $\sum_{ j\in I_j} p_j = T$?
It is easy to see that Partition-into-$k$-sets is NP-complete for every $k\geq 2$, via a polynomial time reduction from \ssc{partition}.
Starting from the \ssc{partition} instance, we add $k-2$ integers each of value $\text{SUM}/2$ where SUM refers to the total value of all elements of the \ssc{partition} instance. 
We have finished the construction of the Partition-into-$k$-sets instance. One can verify that the Partition-into-$k$-sets instance is a ``yes'' instance if and only if the \ssc{partition} instance is a ``yes'' instance.

    We present the reduction from the Partition-into-$(n-1)$-sets problem, known to be NP-complete based on the above argument.
    Given an arbitrary instance of Partition-into-$(n-1)$-sets, we construct a fair division instance $\cI$ with $n$ agents and a set $\{o_1,\ldots,o_{q+r}\}$ of $q+2$ items. The cost functions of agents are shown in Table \ref{tab:ef1-hard-f} where $K$ is a sufficiently large number, i.e., $K\gg T$.
    For any $i\in [n]$, we have $\sum_{j=1}^{q+2}c_i(o_j) = (\frac{n-1}{K}+2)T$, and thus, cost functions are normalized.
       \begin{table}[htbp]
  \centering
  \caption{The constructed fair division instance in Theorem \ref{thm:hard-ef1-sc}}
  \label{tab:ef1-hard-f}
    \begin{tabular}{|c|c|cccccc|}
    \hline
    \multicolumn{2}{|c|}{\multirow{2}[4]{*}{$c_i(o)$}} & \multicolumn{6}{c|}{Items} 
    \\
\cline{3-8}    \multicolumn{2}{|c|}{} & $o_1$ & $ o _2$ & $\cdots$ & $o_q$ & $o_{q+1}$ & $o_{q+2}$ \bigstrut\\
    \hline
    \multirow{4}[2]{*}{Agents} & $1$ & $\frac{p_1}{K}$ & $ \frac{p_2}{K}$ & $\cdots$ & $ \frac{p_q}{K}$ & $  T$ & $ T$   
    \\
          & $2$   & $\frac{p_1}{K}$ & $ \frac{p_2}{K}$ & $\cdots$ & $ \frac{p_q}{K}$ & $ T$ & $ T$   \\
          & $\vdots$ &  $\vdots$ & $\vdots$ & $\ddots$ & $\vdots$ & $\vdots$ & $\vdots$  \\
          & $n-1$ & $\frac{p_1}{K}$ & $\frac{p_2}{K}$ & $\ldots$ & $\frac{p_q}{K}$ & $T$ & $T$ \\
          & $n$ & $\frac{2K+n-1}{Kn+K}p_1$ & $\frac{2K+n-1}{Kn+K}p_2$ & $\cdots$ & $\frac{2K+n-1}{Kn+K}p_q$ & $\frac{2K+n-1}{Kn+K}T$ & $\frac{2K+n-1}{Kn+K}T$   \\
    \hline
    \end{tabular}%
\end{table}%

Suppose we have a ``yes'' instance of Partition-into-$(n-1)$-sets. Let $I_1,\ldots,I_{n-1}$ be a solution. 
We now consider the allocation $A=(A_1,\ldots,A_n)$ where for any $i\in[n-1]$, $A_i=\bigcup_{j\in I_i} o_j$ and $A_n=\{o_{q+1},o_{q+2}\}$.
For every agent $i\in [n-1]$, it is not hard to verify that agent $i$ does not envy other agents. For agent $n$, since $\min_{o\in A_n} c_n(A_n\setminus\{o\}) = \frac{2K+n-1}{Kn+K}T=c_n(A_i)$ for all $i\in [n-1]$, she also satisfies the condition of EF1 (and hence $A$ is EF1).
Moreover allocation $A$ also achieves the minimum social cost $\frac{n-1}{K}T+\frac{4K+2n-2}{Kn+K}T$.

For the other direction, suppose we have ``no'' instance. We first claim that in an EF1 allocation, the bundle of agent $n$ does not include both $o_{q+1}, o_{q+2}$.
Let $\{S_1,\ldots,S_{n-1}\}$ be an arbitrary $(n-1)$-partition of $\{ o_1,\ldots, o_q \}$. Due to ``no'' instance, $\min_{j\in[n-1]} c_n(S_j) < \frac{2K+n-1}{Kn+K} T $ holds. Accordingly, even when agent 3 receives only two items $o_{r+1}, o_{r+2}$, she violates the property of EF1 as
$\min\{c_n(o_{q+1}), c_n(o_{q+2})\} > \min_{j\in[n-1]}c_n(S_j)$.
Thus, in every EF1 allocation, at least one of $ o_{r+1}, o_{r+2} $ is not allocated to agent $n$ and therefore the social cost of an EF1 allocation is at least $\frac{n-1}{K}T + \frac{2K+n-1}{Kn+K}T + T$.

By scaling the total cost of every agent to 1, we can claim that
it is NP-hard to approximate SC-EF1 within an additive factor of
$$
\frac{1-\frac{2K+n-1}{Kn+K}}{\frac{n-1}{K}+2} \rightarrow \frac{n-1}{2n+2} \textnormal{ as } K\rightarrow +\infty,
$$
which yields the desired inapproximation factor.
\end{proof}

For the case of two agents, we also show the hardness of SC-EF1.

\begin{theorem}\label{thm:hard-ef1-sc-n=2}
    When $n=2$, SC-EF1 is NP-hard.
\end{theorem}
\begin{proof}
    It suffices to prove NP-completeness of the decision version of SC-EF1: given a fair division instance and a threshold value $W$, does there exist an EF1 allocation with social cost at most $W$? 
    
    We present the reduction from the \ssc{partition} problem. Given an arbitrary instance of \ssc{partition}, we construct a fair division instance $\cI$ of 2 agents and a set $\{ o_1,\ldots, o_{r+3} \}$ of $r+3$ items. The cost functions of agents are shown in Table \ref{tab:ef1-hard-n=2} where $K$ is a sufficiently large number, i.e., $K\gg T$. The cost functions are normalized as for every $i\in [2]$, $\sum_{j=1}^{r+3}c_i(o_j)=2T+8K$.
    We define the threshold value $W$ as $W = \frac{13}{2}K + 2T$.

    \begin{table}[htbp]
  \centering
  \caption{The constructed fair division instance in Theorem \ref{thm:hard-ef1-sc-n=2}}
  \label{tab:ef1-hard-n=2}
    \begin{tabular}{|c|c|ccccccc|}
    \hline
    \multicolumn{2}{|c|}{\multirow{2}[4]{*}{$c_i(o)$}} & \multicolumn{7}{c|}{Items} 
    \\
\cline{3-9}    \multicolumn{2}{|c|}{} & $o_1$ & $ o _2$ & $\cdots$ & $o_r$ & $o_{r+1}$ & $o_{r+2}$ & $o_{r+3}$ \bigstrut\\
    \hline
    \multirow{2}[2]{*}{Agents} & $1$ & $p_1$ & $ p_2$ & $\cdots$ & $ p_r$ & $  6K$ & $ 2K$   & 0 
    \\
          & $2$   & $\frac{p_1}{2}$ & $ \frac{p_2}{2}$ & $\cdots$ & $ \frac{p_r}{2}$ & $ 5K$ & $ \frac{3K+T}{2}$ & $\frac{3K+T}{2}$ \\
    \hline
    \end{tabular}%
\end{table}%

Suppose we have a ``yes'' \ssc{partition} instance and let $I_1$ and $I_2$ be a solution. We now consider the allocation $A=(A_1, A_2)$ where $A_1 =\{ o_{r+3}\} \cup  \bigcup_{j \in I_1} o_j$ and $A_2 = \{ o_{r+1}, o_{r+2} \} \bigcup_{j \in I_2} o_j $.
As $c_1(A_1) < c_1(A_2)$ and $c_2(A_2\setminus \{o_{r+1} \}) = c_2(A_1)$, allocation $A$ is EF1 and moreover has social cost $\SC(A)= \frac{13}{2}K + 2T = W$.

We now prove the other direction. Suppose we have a ``no'' instance of \ssc{partition}. Let $B=(B_1,B_2)$ be an allocation with $\SC(B) \leq \frac{13}{2}K + 2T$. We will show that allocation $B$ can never be EF1.
Since $\SC(B) \leq \frac{13}{2}K + 2T$, it must hold that $o_{r+3} \in B_1$ and $o_{r+1}, o_{r+2} \in B_2$.
Moreover as $\SC(B) \leq W$, agent 1 can receive a cost of at most $T$ from the assignment of $o_1,\ldots, o_r $.
Due to ``no'' \ssc{partition} instance, we have $c_1(B_1 \cap \bigcup_{j\in[r]} o_j) < T$, implying $c_2(B_1)<\frac{3}{2}K + T$.
Accordingly for agent 2, she receives cost larger than $\frac{T}{2}$ from the assignment of $o_1,\ldots,o_r$, and thus, $c_2(B_2\setminus \{ o_{r+1}\}) > \frac{3}{2}K + T > c_2(B_1)$, violating the condition of EF1.
Therefore, no EF1 allocation can have social cost no greater than $W =\frac{13}{2}K+2T$.
\end{proof}

En-route, we also prove that no polynomial time algorithm guarantees a bounded multiplicative approximation factor.

\begin{theorem}\label{thm:hard-multi-ef1}
    For $n\geq 4$, there is no polynomial time algorithm having a bounded multiplicative approximation factor for SC-EF1 unless P = NP.
\end{theorem}
\begin{proof}
    We present the reduction from the \ssc{partition} problem. We will prove the statement for four agents and then explain how to extend the proof to any $n\geq 4$ agents.
    Given an arbitrary \ssc{partition} instance, we construct a fair division instance $\cI$ of 4 agents and a set $\{o_1,\ldots,o_{r+3} \}$ of $r+3$ items. The cost functions of agents are shown in Table \ref{tab:hard-multi-ef1} where $K$ is a sufficiently large number, i.e., $K \gg T$. Cost functions are normalized as for any $i=\in [4]$, $\sum_{j=1}^{r+3}c_i(o_j)=3K+2T$.

    \begin{table}[htbp]
  \centering
  \caption{The constructed fair division instance in Theorem \ref{thm:hard-multi-ef1}}
  \label{tab:hard-multi-ef1}
    \begin{tabular}{|c|c|ccccccc|}
    \hline
    \multicolumn{2}{|c|}{\multirow{2}[4]{*}{$c_i(o)$}} & \multicolumn{7}{c|}{Items} 
    \\
\cline{3-9}    \multicolumn{2}{|c|}{} & $o_1$ & $ o _2$ & $\cdots$ & $o_r$ & $o_{r+1}$ & $o_{r+2}$ & $o_{r+3}$ 
\\
    \hline
    \multirow{4}[2]{*}{Agents} & $1$ & $p_1$ & $ p_2$ & $\cdots$ & $ p_r$ & $  K$ & $ K$   & $K$ \bigstrut[t]\\
          & $2$   & $p_1$ & $ p_2$ & $\cdots$ & $ p_r$ & $ K$ & $ K$ & $K$ \\
    & $3$ & $p_1$ & $ p_2$ & $\cdots$ & $ p_r$ & $T$ & $T$ & $3K-2T$ \\
    & 4 & $ \frac{K+2T}{r}$ & $ \frac{K+2T}{r}$ & $\cdots$  & $ \frac{K+2T}{r}$ & $K$ & $K$ & $0$ \\
    \hline
    \end{tabular}%
\end{table}%

Suppose we have a ``yes'' \ssc{partition} instance and let $I_1$ and $I_2$ be a solution. We now consider the allocation $A=(A_1, A_2,A_3,A_4)$ where $A_1 = \bigcup_{j\in I_1} o_j$, $A_2 = \bigcup_{j\in I_2} o_j$, $A_3 = \{ o_{r+1}, o_{r+2} \}$ and $A_4 = \{ o_{r+3} \}$.
It is not hard to verify that $A$ is an EF1 allocation and has social cost $\SC(A) = 4T$, the minimum social cost of an allocation.

We now prove the other direction. Suppose we have a ``no'' \ssc{partition} instance. Let $B$ be an EF1 allocation. We will prove that $\SC(B) = \Omega(K)$. For a contradiction, assume $\SC(B) = o(K)$.
Then in allocation $B$, both $o_{r+1},o_{r+2}$ must be assigned to agent 3 and moreover $o_{r+3}$ must be allocated to agent 4.
Thus, we have $B_1\cup B_2 \subseteq \bigcup_{ j \in [r]} o_ j $. Due to ``no'' \ssc{partition} instance, we have $\min\{ c_1(B_1), c_2(B_2) \} < T$, and hence, $\min_{o\in B_3} c_3(B_3\setminus \{o \}) = T > \min \{ c_3(B_1), c_3(B_2) \}$.
Thus, agent 3 does not satisfy EF1. In other words, in an EF1 allocation, the union bundle of agents 1 and 2 should contain at least one item from $\{ o_{r+1}, o_{r+2}, o_{r+3} \}$, implying social cost of at least $K$.

Then it is NP-hard to approximate SC-EF1 with a multiplicative factor of $\frac{K}{4T}$ that tends to $+\infty$ when $K$ goes to $+\infty$. 
To extend to any constant $n\geq 4$ agents, one can add $n-4$ copies of agent 4.
\end{proof}



Next, we consider the case of two agents and present an additive approximation scheme for SC-EF1.
In particular, we will
establish an ``equivalence'' relationship between the allocations of goods and the allocations of chores when there are two agents.
When allocating goods to two agents, an FPTAS (multiplicative approximation) for SW-EF1 exists \cite{DBLP:journals/corr/abs-2205-14296}. We shall utilize their results and propose an additive approximation scheme for SC-EF1 by connecting SW-EF1 to our problem.
Denote by $\OPT^{EF1}$ the social cost of the optimal EF1 allocation. 

\begin{theorem}\label{thm:ef1-ptas}
    Let $\epsilon>0$ be an accuracy parameter.
    Given an instance of SC-EF1 with two agents and $m$ items,
    the algorithm runs in time $poly(m,\frac{1}{\epsilon})$ and computes an EF1 allocation with a social cost of at most $\OPT^{EF1}+\epsilon$.
\end{theorem}
We first introduce terminologies of the goods setting. In the allocations of goods, there is a set of  $O$ of items (yielding positive value) to be allocated to a set of agents $N$. Each agent $i\in N$ is associated with an additive valuation function $v_i$ and $v_i(O)=1$. 
To distinguish between instances for goods and instances for chores, we use $\cI^{+}$ and $A^{+}$ (resp. $\cI^{-}$ and $A^{-}$)  to refer to the instance and the allocation for goods (resp. for chores).
An allocation $A^{+}=(A^{+}_1,\ldots, A^{+}_n)$ of goods is EF1 if for any $i,j\in N$, there exists $o \in A^{+}_j$ such that $v_i(A^{+}_i) \geq v_i(A^{+}_j\setminus \{o\})$. The \emph{social welfare} of allocation $A^{+}$ is the total value of agents. Denote by $\SW(A^+)$ the social welfare of $A^+$ and by SW-EF1 the problem of maximizing social welfare subject to EF1 constraint for goods.

When allocating goods to two agents, an FPTAS (multiplicative approximation) for SW-EF1 exists \cite{DBLP:journals/corr/abs-2205-14296}. We shall utilize their results and propose an additive approximation scheme for SC-EF1 by connecting SW-EF1 to our problem when there are two agents.

\begin{lemma}\label{lem:sc-ef1-ak1}
    Given a goods instance $\cI^{+}=\langle \{1,2\}, O, \{v_i\} \rangle$, suppose $A^{+}=(A^{+}_1,A^{+}_2)$ is an allocation maximizing social welfare subject to (goods) EF1 constraint. Considering the chores instance $\cI^{-}=\langle \{1,2\}, O, \{c_i\} \rangle $ where $c_i(o)=v_i(o),\forall o\in O, \forall i =1,2$,
    allocation $A^{-} = (A^{+}_2, A^{+}_1)$ minimizes social cost subject to (chores) EF1 constraint.
\end{lemma}
\begin{proof}
    Let $B^{+}=(B^+_1,B^+_2)$ be an arbitrary EF1 allocation of $\cI^+$. Construct $B^-=(B^-_1, B^-_2)$, an allocation of instance $\cI^-$, by letting $B^-_1=B^+_2$ and $B^-_2=B^+_1$. Then we have
    $
    \SC(B^-) = c_1(B^-_1)+ c_2(B^-_2)= v_1(B^+_2)+ v_2(B^+_1) = 2-\SW(B^+)
    $; the second equality transition is because $c_i(o)=v_i(o),\forall o\in O, \forall i =1,2$ and the third equality transition is because $v_i(O)=1$ for all $i$.
    We then prove that allocation $B^-$ is also EF1. Since goods allocation $B^+$ is EF1, there exists $o'\in B^+_2$ such that $v_1(B^+_1) \geq v_1(B^+_2\setminus \{o'\})$, equivalent to $c_1(B^+_1) \geq c_1(B^+_2\setminus \{o'\})$ as $c_i(\cdot)=v_i(\cdot)$.
    The latter inequality translates to $c_1(B^-_1\setminus\{o'\}) \leq c_1(B^-_2)$ with $o'\in B^-_1$, indicating that agent 1 satisfies (chores) EF1 in allocation $B^-$. By similar arguments, agent 2 also satisfies EF1 and hence $B^-$ is EF1. 
    We remark that the reverse implication is also true, that is, given an arbitrary EF1 allocation $B^-$ of $\cI^-$, the allocation $B^+$ achieved by swapping bundles of two agents is EF1 and moreover $\SC(B^-)+\SW(B^+) = 2$.

    The above argument proves that allocation $A^-$ is EF1. The remaining is to show that $A^-$ minimizes social cost among all EF1 allocations. For a contradiction, let $D^-=(D^-_1, D^-_2)$ be an EF1 allocation of $\cI^-$ with social cost $\SC(D^-)< \SC(A^-)$.
    Then consider allocation $D^+=(D^-_2, D^-_1)$ of $\cI^+$. By the above argument, $D^+$ is EF1. Note that $\SW(D^+)=2-\SC(D^-)>2-\SC(A^-)=\SC(A^+)$, violating the fact that $A^+$ maximizes social welfare among all EF1 allocations.
    Therefore, allocation $A^-$ minimizes social cost among all EF1 allocations of $\cI^-$.
\end{proof}

Now we are ready to present the additive approximation scheme for SC-EF1. Denote by $\OPT^{EF1}$ the social cost of the optimal EF1 allocation. 


\begin{proof}[Proof of Theorem \ref{thm:ef1-ptas}]
    There is a known FPTAS algorithm for SW-EF1 when $n=2$ \cite{DBLP:journals/corr/abs-2205-14296}, that is, for any instance $\cI^+$ with $m$ items and two agents and for $\epsilon'>0$, their algorithm runs in time $poly(m, \frac{1}{\epsilon'})$ and computes an EF1 allocation $A^+$ with $\SW(A^+) \geq (1-\epsilon')\OPT(\cI^+)$ where $\OPT(\cI^+)$ is the maximum social welfare of EF1 allocations for $\cI^+$.

    Given a chores instance $\cI^-=\langle2,O,\{c_i\} \rangle$, we construct a goods instance $\cI^+=\langle 2, O, \{v_i\} \rangle$ where $v_i(o)=c_i(o),\forall o\in O, \forall i=1,2$.
    Implement the FPTAS in \cite{DBLP:journals/corr/abs-2205-14296} with instance $\cI^+$ and parameter $\epsilon'=\frac{\epsilon}{2}$. Name the returned EF1 allocation $B^+$ that has social welfare $\SW(B^+) \geq (1-\epsilon')\OPT(\cI^+)$.
    Now let us consider allocation $B^-=(B^-_1,B^-_2)$ where $B^-_1=B^+_2$ and $B^-_2=B^+_1$. Then according to the proof of Lemma \ref{lem:sc-ef1-ak1}, allocation $B^-$ is EF1 and has a social cost of $2-\SW(B^+)$.
    Denote by $\OPT(\cI^-)$ the minimum social cost of EF1 allocations for $\cI^-$. Again by the proof of Lemma \ref{lem:sc-ef1-ak1}, we have $\OPT(\cI^-)+\OPT(\cI^+)=2$. Thus, we have the following,
    $$
    \begin{aligned}
    \SC(B^-) = 2-\SW(B^+) &\leq 2-(1-\epsilon')\OPT(\cI^+) \\
    & = 2-\OPT(\cI^+) + \frac{\epsilon}{2} \cdot \OPT(\cI^+) \\
    & \leq \OPT(\cI^-) + \epsilon,
        \end{aligned}
    $$
    where the equality transition is due to $\epsilon'=\frac{\epsilon}{2}$ and the last inequality transition is due to $\OPT(\cI^+) \leq 2$.

    As for the running time, since converting $\cI^-$ to $\cI^+$ takes linear time, the total running time of our algorithm is $poly(m,\frac{1}{\epsilon})$. 
\end{proof}

\section{Conclusion}

In this paper, we re-raised the discussion on the trade-off between fairness and system efficiency for the allocations of indivisible chores. 
We propose a new measurement, the cost of fairness, to quantify the social cost rise when requiring the allocations to be fair.
We found that an EQ1 or EF1 allocation with social cost at most 1
can be computed in polynomial time. 
However, with three or more agents, finding an {EF1 (or EQ1 respectively)} allocation with {social} cost smaller than $1/2$ (or $1/4$ respectively) plus the {social} cost of the optimal EQ1 (or EF1 respectively) allocation is NP-hard.
We complement these results with additive {approximation schemes} 
for the two-agent case. 
Regarding EQX, we prove that there exist instances in which any EQX allocation has social cost $n$ and finding an allocation with cost smaller than $n$ is NP-hard for any $n$. 

Our work uncovers many open problems and future research directions. 
First, in the current paper, we did not consider share-based fairness notions such as PROP1, PROPX, and approximate MMS. The price of PROP1 and PROPX is shown to be $n$ which is acceptable compared with infinity. It is proved in \cite{DBLP:conf/www/0037L022} that there is always a PROP1 and PROPX allocation with social cost of at most 1, which implies that the cost of PROP1 and PROPX is 1.
We leave a complete study of share-based fairness notions for future study.
Secondly, we only focused on the allocation of chores.
It is interesting to see whether the cost of fairness brings new insights to the trade-off between fairness and efficiency for goods. 

\section*{Acknowledgements}
This work is funded by the Hong Kong SAR Research Grants Council (No. PolyU 15224823) and the Guangdong Basic and Applied Basic Research Foundation (No. 2023A1515010592).

\newpage

\bibliographystyle{alpha}
\bibliography{sample.bib}

\end{document}